%% file: paper.tex
\documentclass[letterpaper,twocolumn,10pt]{article}
\usepackage{graphicx, xspace, outlines,url,xcolor}
\usepackage[compact]{titlesec}

\titleformat{\subsection}{\normalfont\fontsize{10.5}{13}\bfseries}{\thesubsection}{1em}{}

\usepackage{usenix,epsfig,endnotes}
\usepackage{amsfonts}
\usepackage[square, comma, numbers,sort&compress]{natbib}
\usepackage{amsmath, amsthm, amssymb}

\newtheorem{theorem}{Theorem}

\newtheorem{lemma}{Lemma}

\usepackage[pdfstartview=FitH, bookmarksnumbered=true, bookmarksopen=true, colorlinks=true, citecolor=blue, linkcolor= blue, urlcolor=blue]{hyperref}
\usepackage{soul}
\usepackage{url}
\usepackage[hang,scriptsize,tight,nooneline]{subfigure}

\newcommand{\cut}[1]{}

\newcommand{\ie}{{\it i.e.},\xspace}

\newcommand{\cutatlastminute}[1]{}

\newcommand{\paragraphb}[1]{\vspace{0.1in}\noindent{\bf #1} }

\newcommand{\ankitblue}{\textcolor{black}}

\newcommand{\ankitr}{\textcolor{black}}

\newcommand{\apl}{\langle D \rangle}

\setstcolor{red}

\begin{document}
\baselineskip=12bp

\title{High Throughput Data Center Topology Design}

\author{
Ankit Singla, P. Brighten Godfrey, Alexandra Kolla\\
University of Illinois at Urbana--Champaign\\
}

\maketitle

\input{abstract}
\vspace{-4pt}
\input{intro}
\vspace{-4pt}
\input{related}
\vspace{-4pt}
\input{methodology}
\vspace{-4pt}
\input{randoptimal}
\vspace{-4pt}
\input{serverdist}

\input{bias}
\input{linespeed}
\vspace{-4pt}
\input{cause}
\input{analysis_new}

\vspace{-4pt}
\input{comparison}

\vspace{-4pt}
\input{practice}
\vspace{-6pt}
\input{discussion}
\vspace{-6pt}
\input{conclusion}
\vspace{-6pt}
\input{ack}

{\bibliographystyle{abbrv}
\setlength{\bibsep}{0pt}
{
\bibliography{paper}
}
\end{document}

%% file: abstract.tex
\abstract{\vspace{3pt}
With high throughput networks acquiring a crucial role in supporting data-intensive applications, a variety of data center network topologies have been proposed to achieve high capacity at low cost. While this work explores a large number of design points, even in the limited case of a network of identical switches, no proposal has been able to claim any notion of \emph{optimality}. The case of \emph{heterogeneous} networks, incorporating multiple line-speeds and port-counts as data centers grow over time, introduces even greater complexity.

In this paper, we present the first non-trivial upper-bound on network throughput under uniform traffic patterns for \emph{any} topology with identical switches. We then show that random graphs achieve throughput surprisingly close to this bound, within a few percent at the scale of a few thousand servers. Apart from demonstrating that homogeneous topology design may be reaching its limits, this result also motivates our use of random graphs as building blocks for design of heterogeneous networks. Given a heterogeneous pool of network switches, we explore through experiments and analysis, how the distribution of servers across switches and the interconnection of switches affect network throughput. We apply these insights to a real-world heterogeneous data center topology, VL2, demonstrating as much as $43\%$ higher throughput with the same equipment.
}

%% file: intro.tex
\section{Introduction}
\label{sec:intro}

Data centers are playing a crucial role in the rise of Internet services 
and big data. In turn, efficient data center operations depend on high capacity
networks to ensure that computations are not bottlenecked on communication.
As a result, the problem of designing massive high-capacity network
interconnects has become more important than ever. Numerous data center 
network architectures have been proposed in response to this need~\cite{fattree, 
vl2, bcube, portland, dcell, cthrough, helios, proteus, mirror-mirror, jellyfish,
legup, rewire}, exploiting a variety of network topologies to achieve high throughput, 
ranging from fat trees and other Clos networks~\cite{fattree,vl2} 
to modified generalized hypercubes~\cite{bcube} to small world networks~\cite{swdc} 
and uniform random graphs~\cite{jellyfish}.



\ankitblue{However, while this extensive literature exposes several points in the topology design space,
even in the limited case of a network of identical switches, it does not answer a fundamental question:
\emph{How far are we from throughput-optimal topology design?}}
The case of \emph{heterogeneous} networks, \ie networks composed of switches or servers
with disparate capabilities, introduces even greater complexity.  
Heterogeneous network equipment is, in fact, the common case in the typical data
center: servers connect to top-of-rack (ToR) switches, which
connect to aggregation switches, which connect to core switches, with each
type of switch possibly having a different number of ports as well some
variations in line-speed. For instance, the ToRs may have both $1$ Gbps and 
$10$ Gbps connections while the rest of the network may have only $10$ Gbps links.
Further, as the network expands over the years and new, more powerful
equipment is added to the data center, one can expect more heterogeneity ---
each year the number of ports supported by non-blocking commodity 
Ethernet switches increases. While line-speed changes are slower, the 
move to $10$ Gbps and even $40$ Gbps is happening now, and higher line-speeds
are expected in the near future.

In spite of heterogeneity being commonplace in data center networks, very little 
is known about heterogeneous network design. For instance, there is no clarity on whether the 
traditional ToR-aggregation-core organization is superior to a ``flatter'' network without such a switch hierarchy;
or on whether powerful core switches should be connected densely together, or spread more evenly throughout the network.

\smallskip{}

\ankitblue{\emph{The goal of this paper is to develop an understanding of how to design high throughput network topologies at limited cost, even when heterogeneous components are involved, and to apply this understanding to improve real-world data center networks.}}
This is nontrivial: Network topology design is hard in general, because of the combinatorial 
explosion of the number of possible networks with size. Consider, for example, the related\footnote{Designing for low network diameter is related to designing for
high throughput, because shorter path lengths translate to the network using less capacity
to deliver each packet; see discussion in~\cite{jellyfish}.}
\emph{degree-diameter problem}~\cite{degree-diameter}, a well-known graph theory problem
where the quest is to pack the largest possible number of nodes 
into a graph while adhering to constraints on both the degree and the diameter. Non-trivial optimal solutions are known for a total of only seven combinations of degree 
and diameter values, and the largest of these optimal networks has only $50$ nodes! The lack of symmetry
that heterogeneity introduces only makes these design problems more challenging.

To attack this problem, we decompose it into several steps which together give a high level understanding of network topology design, and yield benefits to real-world data center network architectures.  First, we address the case of networks of homogeneous servers and switches.  Second, we study the heterogeneous case, optimizing the distribution of servers across different classes of switches, and the pattern of interconnection of switches.  Finally, we apply our understanding to a deployed data center network topology.  Following this approach, our key results are as follows.

\cut{On the other hand, for the homogeneous case, Jellyfish is within a 
few percent of the upper bound on throughput for \emph{any} network built with the same equipment~\cite{topocomp}.
Thus, in our investigation of heterogeneous topology design, we leverage Jellyfish's
success with random graphs, and couple it with our decomposition of the problem
into the following smaller pieces which, put together, yield a complete picture: (a)
Deciding on a distribution of servers across switches; (b) Deciding on an interconnection
of switches; (c) Deciding on what switches to purchase given multiple options and a budget;
and (d) Informing the design of switches for them to maximally benefit the network. \cut{We note that
these decisions are not independent -- for instance, the server distribution clearly depends
on how many ports on the each switch have been used in the switch interconnection.}
Following this approach, the contributions of this work are as follows.}

\paragraphb{(1) Near-optimal topologies for homogeneous networks.} We present an upper bound on network throughput for any topology with identical switches, as a function of the number of switches and their degree (number of ports).  Although designing \emph{optimal} topologies is infeasible, we demonstrate that random graphs achieve throughput surprisingly close to this bound---within a few percent at the scale of a few thousand servers for random permutation traffic.  This is particularly surprising in light of the much larger gap between bounds and known graphs in the related degree-diameter problem~\cite{degree-diameter}\footnote{For instance, for degree $5$ and diameter $4$, the best known graph has only $50\%$ of the number of nodes in the best known upper bound~\cite{degdiaGap}. Further, this gap grows larger with both degree and diameter.}.


We caution the reader against over-simplifying this result to `flatter topologies are better': Not all `flat' or `direct-connect' topologies (where all switches connect to servers) perform equally. For example, random graphs have roughly $30\%$ higher throughput than hypercubes at the scale of $512$ nodes, and this gap increases with scale~\cite{sigmetricsdraft}. Further, the notion of `flat' is not even well-defined for heterogeneous networks.

\paragraphb{(2) High-throughput heterogeneous network design.} We use random graphs as building blocks for heterogeneous network design by first optimizing the volume of connectivity between groups of nodes, and then forming connections randomly within these volume constraints. Specifically, we first show empirically that in this framework, for a set of switches with different port counts but uniform line-speed, attaching servers to switches in proportion to the switch port count is optimal.

Next, we address the interconnection of multiple types of switches. For tractability, we limit our investigation to two switch types.  Somewhat surprisingly, we find that a wide range of connectivity arrangements provides nearly identical throughput. A useful consequence of this result is that there is significant opportunity for clustering switches to achieve shorter cable lengths on average, without compromising on throughput. Jellyfish~\cite{jellyfish} demonstrated this experimentally. Our results provide the theoretical underpinnings of such an approach.

Finally, in the case of multiple line-speeds, we show that complex bottleneck behavior may appear and there may be multiple configurations of equally high capacity. 

\paragraphb{(3) Applications to real-world network design.} The topology proposed in VL2~\cite{vl2} incorporates heterogeneous line-speeds and port-counts, and has been deployed in Microsoft's cloud data centers.\footnote{Based on personal exchange, and mentioned publicly at \url{http://research.microsoft.com/en-us/um/people/sudipta/}.} We show that using a combination of the above insights, VL2's throughput can be improved by as much as $43\%$ at the scale of a few thousand servers simply by rewiring existing equipment, with gains increasing with network size.  

\medskip
While a detailed treatment of other related work follows in \S\ref{sec:related}, 
the \textbf{Jellyfish}~\cite{jellyfish} proposal merits attention here since it is also
based on random graphs. Despite this shared ground, Jellyfish does not address either of the central questions addressed
by our work: (a) How close to optimal are random graphs for the homogeneous case? and
(b) How do we network \emph{heterogeneous} equipment for high throughput?
In addition, unlike Jellyfish, by analyzing how network metrics like cut-size, path length, 
and utilization impact throughput, we attempt to develop an \emph{understanding} of network 
design.

\cut{We note that while the applications in this paper concern data centers, the question of network design for high capacity is fundamental, and is likely to be of independent interest.}

%% file: related.tex
\section{Background and Related Work}
\label{sec:related}

High capacity has been a core goal of communication networks since their inception.  How that goal manifests in network topology, however, has changed with systems considerations. Wide-area networks are driven by geographic constraints such as the location of cities and railroads.  Perhaps the first high-throughput networks not driven by geography came in the early 1900s. To interconnect telephone lines at a single site such as a telephone exchange, \emph{nonblocking} switches were developed which could match inputs to any permutation of outputs.  Beginning with the basic crossbar switch which requires $\Theta(n^2)$ size to interconnect $n$ inputs and outputs, these designs were optimized to scale to larger size, culminating with the Clos network developed at Bell Labs in 1953~\cite{clos1953study} which constructs a nonblocking interconnect out of $\Theta(n \log n)$ constant-size crossbars.

In the 1980s, supercomputer systems began to reach a scale of parallelism for which the topology connecting compute nodes was critical.  Since a packet in a supercomputer is often a low-latency memory reference (as opposed to a relatively heavyweight TCP connection) traversing nodes with tiny forwarding tables, such systems were constrained by the need for very simple, loss-free and deadlock-free routing.  As a result the series of designs developed through the 1990s have simple and regular structure, some based on non-blocking Clos networks and others turning to  butterfly, hypercube, 3D torus, 2D mesh, and other designs~\cite{leighton}.

In commodity compute clusters, increasing parallelism, bandwidth-intensive big data applications and cloud computing have driven a surge in data center network architecture research. An influential $2008$ paper of Al-Fares et al.~\cite{fattree} proposed moving from a traditional data center design utilizing expensive core and aggregation switches, to a network built of small components which nevertheless achieved high throughput --- a folded Clos or ``fat-tree'' network. This work was followed by several related designs including Portland~\cite{portland} and VL2~\cite{vl2}, a design based on small-world networks~\cite{swdc}, designs using servers for forwarding~\cite{bcube, dcell, mdcube}, and designs incorporating optical switches~\cite{helios,cthrough}.

Jellyfish~\cite{jellyfish} demonstrated, however, that Clos networks are sub-optimal.  In particular, \cite{jellyfish} constructed a random degree-bounded graph among switch-to-switch links, and showed roughly $25\%$ greater throughput than a fat-tree built with the same switch equipment.\cut{, as well as higher throughput than small-world-based topologies~\cite{swdc}. } In addition, \cite{jellyfish} showed quantitatively that random networks are easier to incrementally expand --- adding equipment simply involves a few random link swaps. Several challenges arise with building a completely unstructured network; \cite{jellyfish} demonstrated effective routing and congestion control mechanisms, and showed that cable optimizations for random graphs can make cable costs similar to an optimized fat-tree while still obtaining substantially higher throughput than a fat-tree.

\cut{Our work is closely related to Jellyfish in that we use random networks as a foundation.  However, \cite{jellyfish} did not address the question of how close random networks are to the optimal topology,\footnote{The Jellyfish paper did compare with several of the best known degree-diameter networks, but while it is suggestive, this is not a comparison with \emph{throughput-optimal} networks.} and studied only homogeneous networks.}

While the literature on homogeneous network design is sizeable, very little is known about heterogeneous topology design, perhaps because earlier supercomputer topologies (which reappeared in many recent data center proposals) were generally constrained to be homogeneous.  VL2~\cite{vl2} provides a point design, using multiple line-speeds and port counts at different layers of its hierarchy; we compare with VL2 later (\S\ref{sec:vl2}).
The only two other proposals that address heterogeneity are 
LEGUP~\cite{legup} and REWIRE~\cite{rewire}. LEGUP uses an optimization framework
to search for the cheapest Clos network achieving
desired network performance. Being restricted
to Clos networks impairs LEGUP greatly: Jellyfish achieves the same network
expansion as LEGUP at $60\%$ lower cost~\cite{jellyfish}. REWIRE removes this
restriction by using a local-search optimization (over a period of several days of compute time at the scale of $3200$ servers) to continually improve upon an 
initial feasible network.
REWIRE's code is not available so a comparison has not been possible.  But more fundamentally, all of the above approaches are either point designs~\cite{vl2} or heuristics~\cite{legup, rewire} which by their blackbox nature, provide neither
an \emph{understanding} of the solution space, nor any evidence of near-optimality.

\cut{While REWIRE's code not being available makes a quantitative
comparison hard, we make the following arguments in favor of our approach: (a) at the
least, we shall provide a \emph{good} initial network for REWIRE to optimize -- given that
REWIRE's optimization can take days\footnote{The largest network tested in REWIRE has $3200$ servers. 
All results in the paper are from $72$ hours of REWIRE-optimization.} even for 
a few thousand servers this in itself is significant; (b) as
our parallel work~\cite{topocomp} shows, there is very little room for improvement beyond 
the random graph -- even at the scale of a few thousand servers, it is within 
a few percent of the \emph{upper bound} on the capacity of any graph built 
with identical equipment; (c) our work yields insights into the structure
of an optimized network, as opposed to REWIRE's opacity; and (d) we address new and
interesting questions, including the distribution of servers across
switches, and rethinking the provisioning of ports on switches to enable the construction 
of higher capacity networks.}

%% file: methodology.tex
\section{Simulation Methodology}
\label{sec:method}

Our experiments measure the capacity of network topologies. For most of this paper, our goal is to study topologies explicitly independent of systems-level issues such as routing and congestion control. Thus, we model network traffic using fluid splittable flows which are routed optimally. Throughput is then the solution to the standard maximum concurrent multi-commodity flow problem~\cite{leightonmaxflow}. Note that by maximizing the minimum flow throughput, this model incorporates a strict definition of fairness. We use the CPLEX linear program solver~\cite{cplex} to obtain the maximum flow.  Unless otherwise specified, the workload we use is a random permutation traffic matrix, where each server sends traffic to (and receives traffic from) exactly one other server.


In \S\ref{sec:practice}, we revisit these assumptions to address systems concerns.
We include results for several other traffic matrices besides permutations. We also show that throughput within a few percent of the optimal flow values from CPLEX can be achieved \emph{after} accounting for packet-level routing and congestion control inefficiencies.

Any comparisons between networks are made using identical switching equipment, unless noted otherwise. 

Across all experiments, we test a wide range of parameters, varying the network size, node degree, and oversubscription. A representative sample of results is included here. Most experiments average results across $20$ runs, with standard deviations in throughput being $\sim$$1\%$ of the mean except at small values of throughput in the uninteresting cases. Exceptions are noted in the text.

Our simulation tools are publicly available~\cite{codelink}.

%% file: randoptimal.tex
\section{Homogeneous Topology Design}
\label{sec:jellyfish}

\begin{figure}[t]
\centering
\subfigure[]{ \label{fig:optimal:thput}\includegraphics[width=2.2in]{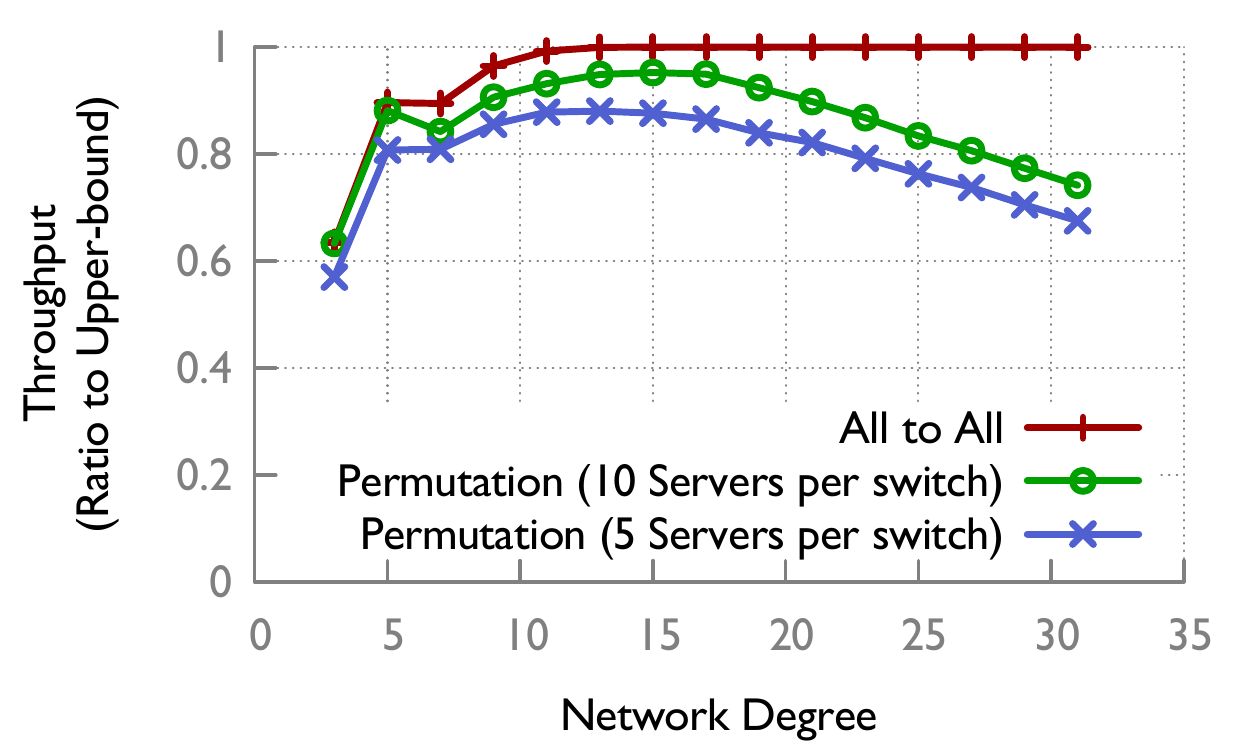}}
\vspace{-8pt}
\subfigure[]{ \label{fig:optimal:pl}\includegraphics[width=2.2in]{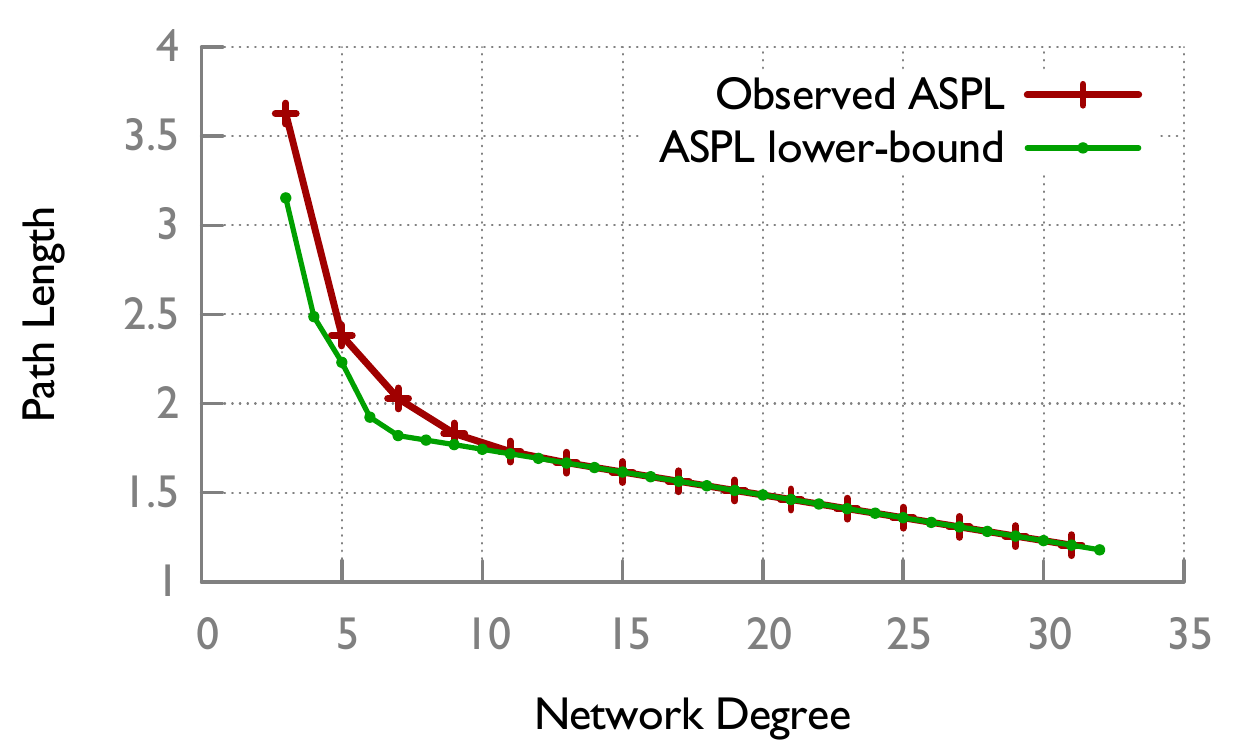}}
\caption{\small \em Random graphs versus the bounds: (a) Throughput and (b) average shortest path length (ASPL) in random regular graphs compared to the respective 
upper and lower bounds for any graph of the same size and degree. The number of switches is fixed to $40$ throughout. The network 
becomes denser rightward on the x-axis as the degree increases.}
\label{fig:optimal}
\vspace{-8pt}
\end{figure}

\begin{figure}[t]
\centering
\subfigure[]{ \label{fig:optimal:thput2}\includegraphics[width=2.2in]{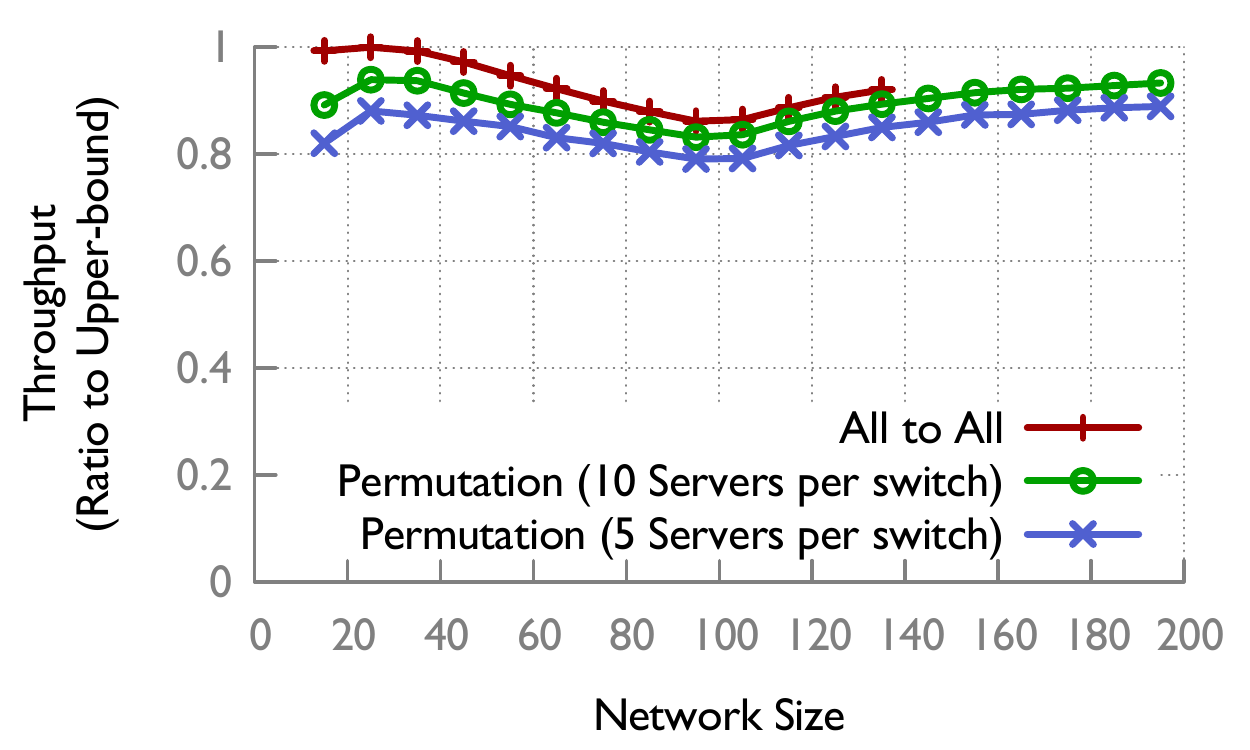}}
\vspace{-8pt}
\subfigure[]{ \label{fig:optimal:pl2}\includegraphics[width=2.2in]{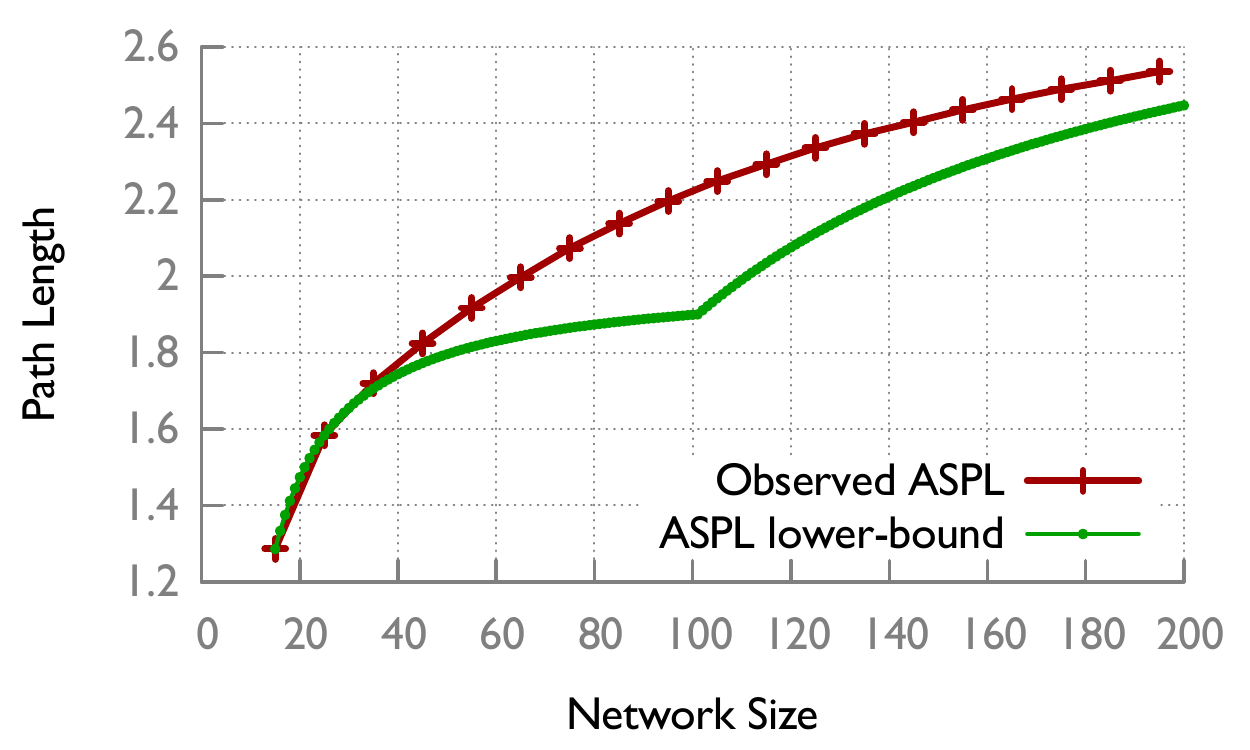}}
\caption{\small \em Random graphs versus the bounds: (a) Throughput and (b) average shortest path length (ASPL) in random regular graphs compared to the respective 
upper and lower bounds for any graph of the same size and degree. The degree is fixed to $10$ throughout. The network 
becomes sparser rightward on the x-axis as the number of nodes increases.}
\label{fig:optimal2}
\vspace{-8pt}
\end{figure}

In this setting, we have $N$ switches, each with
$k$ ports. The network is required to support $S$ servers. The
symmetry of the problem suggests that each switch be connected to
the same number of servers. (We assume for convenience that $S$ is divisible by $N$.) 
Intuitively, spreading servers across switches in a manner that deviates 
from uniformity will create bottlenecks at the switches with larger 
numbers of servers. Thus, we assume that each switch uses out of its $k$ 
ports, $r$ ports to connect to other switches, 
and $k-r$ ports for servers. It is also assumed that each network edge is of
unit capacity.

The design space for such networks is the set of
all subgraphs $H$ of the complete graph over $N$ nodes $K_N$, such that 
$H$ has degree $r$. For generic, application-oblivious design, 
we assume that the objective is to maximize throughput under
a uniform traffic matrix such as all-to-all traffic or random permutation 
traffic among servers. To account for fairness, the network's throughput 
is defined as the maximum value of the minimum flow between 
source-destination pairs. We denote such a throughput 
measurement of an $r$-regular subgraph $H$ of $K_N$ under uniform traffic 
with $f$ flows by $T_{H}(N, r, f)$.  The average path length of the network is denoted by $\apl$.

For this scenario, we prove a simple upper bound on the throughput achievable by 
\emph{any} hypothetical network.

\begin{theorem}\label{thm:thput_pl}
$T_{H}(N, r, f) \leq \frac{Nr}{\apl f}$.
\end{theorem}

\begin{proof}
The network has a total of $Nr$ edges (counting both directions) of unit capacity, 
for a total capacity of $Nr$. A flow $i$ whose end points are a shortest
path distance $d_i$ apart, consumes at least $x_i d_i$ units of capacity in
to obtain throughput $x_i$. Thus, the total capacity consumed by all
flows is at least $\displaystyle\sum\limits_i x_i d_i$. Given that we defined
network throughput $T_{H}(N, r, f)$ as the minimum flow throughput, $\forall i, x_i \geq T_{H}(N, r, f)$.
Total capacity consumed is then at least $T_{H}(N, r, f) \displaystyle\sum\limits_i d_i$.
For uniform traffic patterns such as random permutations and all-to-all traffic, 
$\displaystyle\sum\limits_i d_i = \langle D \rangle f$ because the average source-destination
distance is the same as the graph's average shortest path distance.
Also, total capacity consumed cannot exceed the network's capacity. Therefore,
$\langle D \rangle f T_{H}(N, r, f) \leq Nr$, rearranging which yields the result.
\end{proof}

Further,~\cite{cerfPaper} proves a lower bound on the average shortest path length
of any $r$-regular network of size $N$:

\vspace{-12pt}
\begin{align*}
\langle D \rangle \geq d^{*} &= \frac{\displaystyle\sum\limits_{j=1}^{k-1} jr(r-1)^{j-1} + kR}{N-1}\\
\text{where~~~} R &= N-1 - \displaystyle\sum\limits_{j=1}^{k-1} r(r-1)^{j-1} \geq 0
\end{align*}
and $k$ is the largest integer such that the \emph{inequality} holds.

This result, together with Theorem~\ref{thm:thput_pl}, yields an upper bound on throughput: $T_{H}(N, r, f) \leq \frac{Nr}{f d^*}$.
Next, we show experimentally that random regular graphs
achieve throughput close to this bound.

A {\em random regular graph}, denoted as RRG($N$, $k$, $r$), is a graph sampled 
uniform-randomly from the space of all $r$-regular graphs. This is a well-known construct in graph theory.
As Jellyfish~\cite{jellyfish} showed, RRGs compare favorably against traditional fat-tree topologies, supporting
a larger number of servers at full throughput. However, that fact leaves open the possibility that there
are network topologies that achieve significantly higher throughput than even RRGs. Through experiments, we compare
the throughput RRGs achieve to the upper bound we derived above, and find that our results eliminate this possibility.

Fig.~\ref{fig:optimal:thput} and Fig.~\ref{fig:optimal:thput2} compare throughput achieved by RRGs to the 
upper bound on throughput for any topology built with the same equipment. Fig.~\ref{fig:optimal:thput} shows
this comparison for networks of increasing density (\ie the degree $r$ increases, while the number of nodes $N$
remains fixed at $40$) for $3$ uniform traffic matrices: a random permutation among servers with $5$ servers 
at each switch, another with $10$ servers at each switch, and an all-to-all traffic matrix.
For the high-density traffic pattern, \ie all-to-all traffic, \emph{exact optimal} throughput is achieved by the random
graph for degree $r\geq13$. 
Fig.~\ref{fig:optimal:thput2} shows a similar comparison for increasing size $N$, 
with $r=10$. Our simulator does not scale for all-to-all traffic because the number
of commodities in the flow problem increases as the square of the network size for this pattern. 
Fig.~\ref{fig:optimal:pl} and~\ref{fig:optimal:pl2} compare average shortest
path length in RRGs to its lower bound. 
For both large
network sizes, and very high network density, RRGs are surprisingly
close to the bounds (right side of both figures).

\begin{figure}
\centering
\includegraphics[width=2.4in]{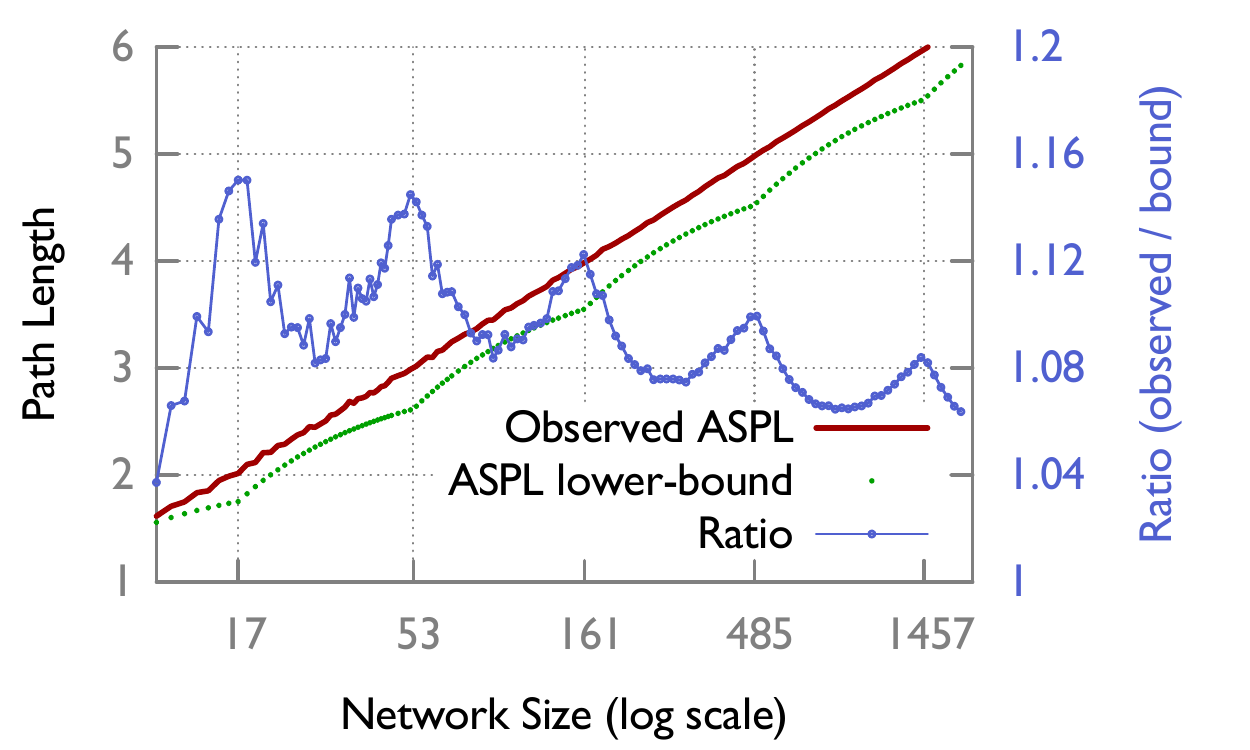}
\vspace{-6pt}
\caption{\small \em ASPL in random graphs compared to the lower bound. The degree is fixed to $4$ throughout. The bound shows a ``curved step" behavior. In addition, as the network size increases, the ratio of observed ASPL to the lower bound approaches $1$. The x-tics correspond to the points where the bound begins new distance levels.}
\label{fig:boundcompare1}
\vspace{-8pt}
\end{figure}

The curve in Fig.~\ref{fig:optimal:pl2} has two interesting features.  First, there is a ``curved step" behavior, with the first step at network size up to $N=101$, and the second step beginning thereafter. To see why this occurs, observe that the bound uses a tree-view of distances from any node --- for a network with degree $d$, $d$ nodes are assumed to be at distance $1$, $d(d-1)$ at distance $2$, $d(d-1)^2$ at distance $3$, etc. While this structure minimizes path lengths, it is optimistic --- in general, not all edges from nodes at distance $k$ can lead outward to unique new nodes\footnote{In fact, prior work shows that graphs with this structure do not exist for $d\geq3$ and diameter $D\geq3$~\cite{degDiaSurvey}.}. As the number of nodes $N$ increases, at some point the lowest level of this hypothetical tree becomes full, and a new level begins.  These new nodes are more distant, so average path length suddenly increases more rapidly, corresponding to a new ``step" in the bound. A second feature is that as $N\to\infty$, the ratio of observed ASPL to the lower bound approaches $1$. This can be shown analytically by dividing an upper bound on the random regular graph's diameter~\cite{diabound} (which also upper-bounds its ASPL) by the lower bound of~\cite{cerfPaper}. For greater clarity, we show in Fig.~\ref{fig:boundcompare1} similar behavior for degree $d=4$, which makes it easier to show many ``steps". \cut{We also note that the upper bound on the random graph's diameter given in Jellyfish~\cite{jellyfish}, which is trivially, also an upper bound on its ASPL, asymptotically matches the lower bound given above.}

The near-optimality of random graphs demonstrated here leads us to use them as a building block for the more complicated case of heterogeneous topology design.

%% file: serverdist.tex
\begin{figure*}[t]
\centering
\subfigure[]{ \label{fig:serverdist:portproportion}\includegraphics[width=2.11in]{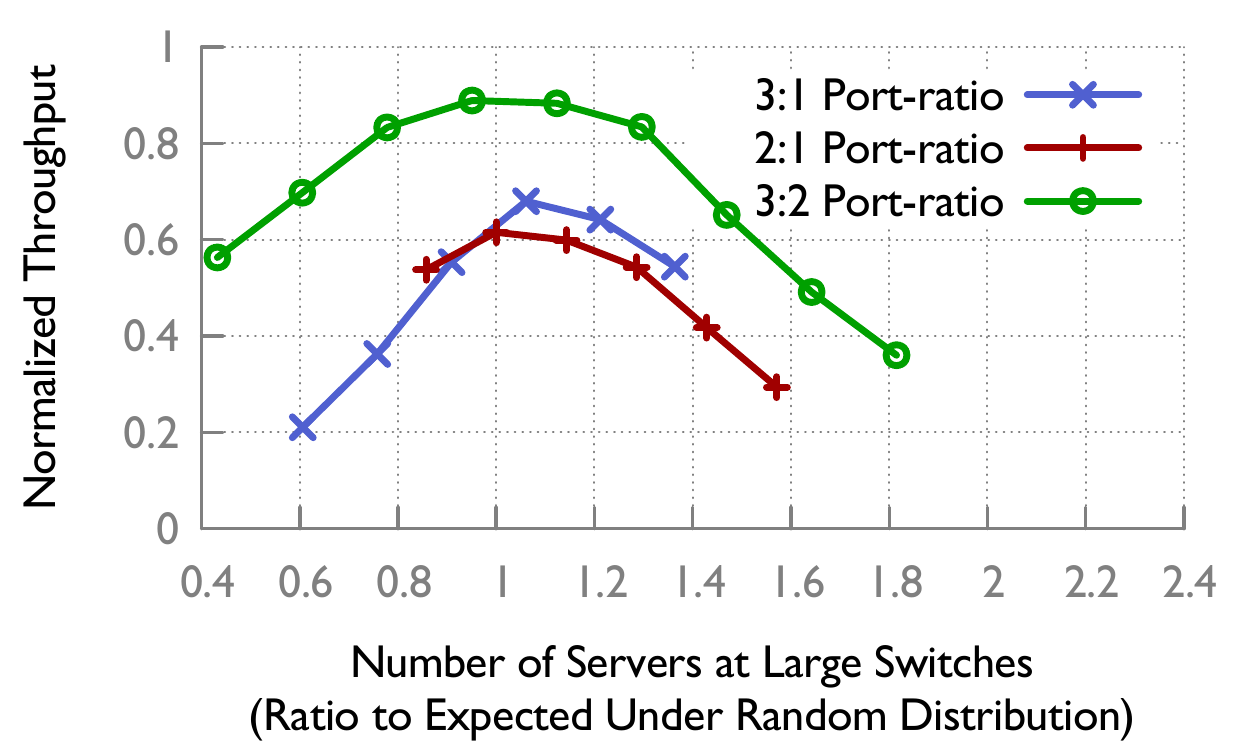}}
\subfigure[]{ \label{fig:serverdist:switchproportion}\includegraphics[width=2.11in]{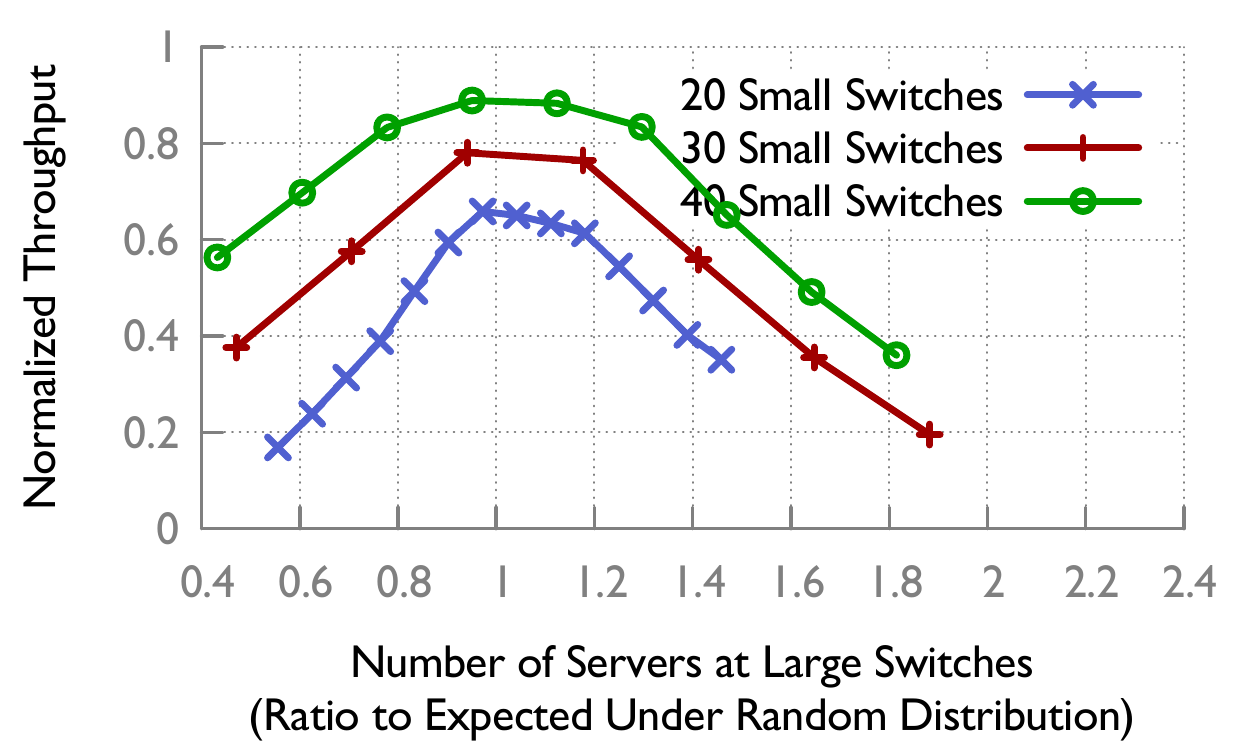}}
\subfigure[]{ \label{fig:serverdist:oversub}\includegraphics[width=2.11in]{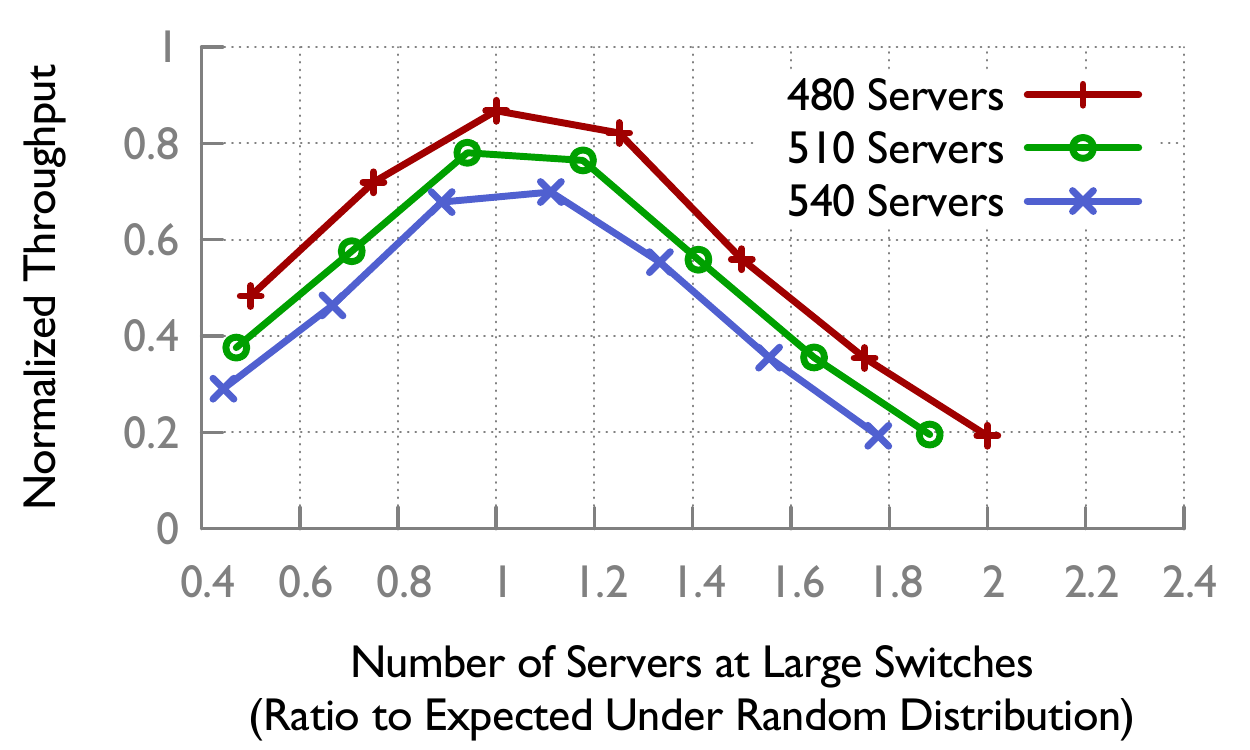}}
\caption{\small \em Distributing servers across switches: Peak throughput is achieved when servers are distributed proportionally to port 
counts \ie x-axis=1, regardless of (a) the absolute port 
counts of switches; (b) the absolute counts of switches of each type; and (c) oversubscription in
the network.}
\label{fig:serverdist}
\end{figure*}

\section{Heterogeneous Topology Design}

With the possible exception of a scenario where a new data center is being built
from scratch, it is unreasonable to expect deployments to have the same, homogeneous
networking equipment. Even in the `greenfield' setting, networks may potentially
use heterogeneous equipment. While our results above show that random graphs achieve 
close to the best possible throughput in the homogeneous network design
setting, we are unable, at present, to make a similar claim for heterogeneous networks,
where node degrees and line-speeds may be different. However, in this section, we 
present for this setting, interesting experimental results which challenge traditional 
topology design assumptions. Our discussion here is mostly limited to the scenario where
there are two kinds of switches in the network; generalizing our results for higher
diversity is left to future work.

\subsection{Heterogeneous Port Counts}
\label{subsec:heterports}

We consider a simple scenario where the network is composed of two types of switches with different 
port counts (line-speeds being uniform throughout). Two natural questions arise that we shall explore
here: (a) How should we distribute servers across the two switch types to maximize
throughput? (b) Does biasing the topology in favor of more connectivity between larger 
switches increase throughput?

First, we shall assume that the interconnection is an unbiased random graph built over the remaining connectivity
at the switches after we distribute the servers. Later, we shall fix the server distribution but
bias the random graph's construction. Finally we will examine the combined effect of varying both
parameters at once.

\paragraphb{Distributing servers across switches:} We vary the numbers of servers apportioned to large and small switches, 
while keeping the total number of servers and switches the same\footnote{Clearly, across the same type of 
switches, a non-uniform server-distribution will cause bottlenecks and sub-optimal throughput.}.
We then build a random graph over the ports that remain unused after attaching the servers. We repeat this 
exercise for several parameter settings, varying the numbers of switches, ports, and servers. A representative 
sample of results is shown in Fig.~\ref{fig:serverdist}. The particular configuration in 
Fig.~\ref{fig:serverdist:portproportion} uses $20$ larger and $40$ smaller switches, with the port counts for the three
curves in the figure being $30$ and $10$ ($3$:$1$), $30$ and $15$ ($2$:$1$), and $30$ and $20$ ($3$:$2$) 
respectively. Fig.~\ref{fig:serverdist:switchproportion} uses $20$ larger switches ($30$ ports) and
$20$, $30$ and $40$ smaller switches ($20$ ports) respectively for its three curves. Fig.~\ref{fig:serverdist:oversub}
uses the same switching equipment throughout: $20$ larger switches ($30$ ports) and $30$ smaller switches ($20$ ports),
with $480$, $510$, and $540$ servers attached to the network. Along the $x$-axis in each figure, the number of
servers apportioned to the larger switches increases. The $x$-axis label normalizes this number to the \emph{expected} number
of servers that would be apportioned to large switches if servers were spread randomly across all the ports in
the network. As the results show, distributing servers in proportion to switch degrees (\ie $x$-axis$=1$) is optimal. 

\begin{figure}
\centering
\includegraphics[width=2.2in]{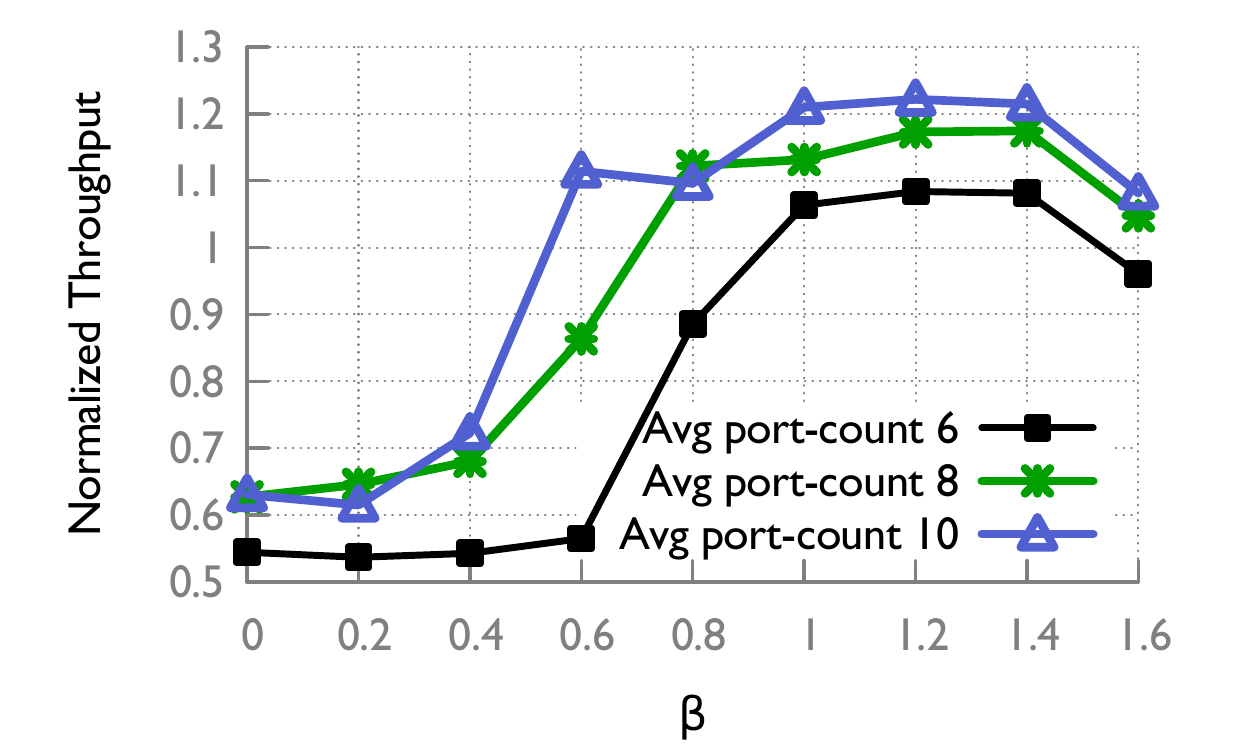}
\caption{\small \em Distributing servers across switches: Switches have port-counts 
distributed in a power-law distribution. Servers are distributed in proportion to the $\beta^{th}$ power of switch port-count.
Distributing servers in proportion to degree ($\beta = 1$) is still among the optimal configurations. \cut{($N$$=$$120$ switches.)}}
\label{fig:powerlaw}
\end{figure}

This result, while simple, is remarkable in the light of current topology design practices, where
top-of-rack switches are the only ones connected directly to servers. 

Next, we conduct an experiment with a diverse set of switch types, rather than just two.
We use a set of switches such that their port-counts $k_i$ follow a power 
law distribution. We attach servers at each switch $i$ in proportion to $k_{i}^{\beta}$, using the remaining
ports for the network. The total number of servers is kept constant as we test various values of $\beta$.
(Appropriate distribution of servers is applied by rounding where necessary to achieve this.) $\beta=0$ implies
that each switch gets the same number of servers regardless of port count, while $\beta=1$ is the same as
port-count-proportional distribution, which was optimal in the previous experiment. The results are shown in Fig.~\ref{fig:powerlaw}. $\beta=1$ is 
optimal (within the variance in our data), but so are other values of $\beta$ such as $1.2$ and $1.4$.
The variation in throughput is large at both extremes of the plot, with the standard deviation being as much
as $10\%$ of the mean, while for $\beta \in \{1, 1.2, 1.4\}$ it is $<4\%$.

%% file: bias.tex
\begin{figure*}[t]
\centering
\subfigure[]{ \label{fig:hhbias:portproportion}\includegraphics[width=2.11in]{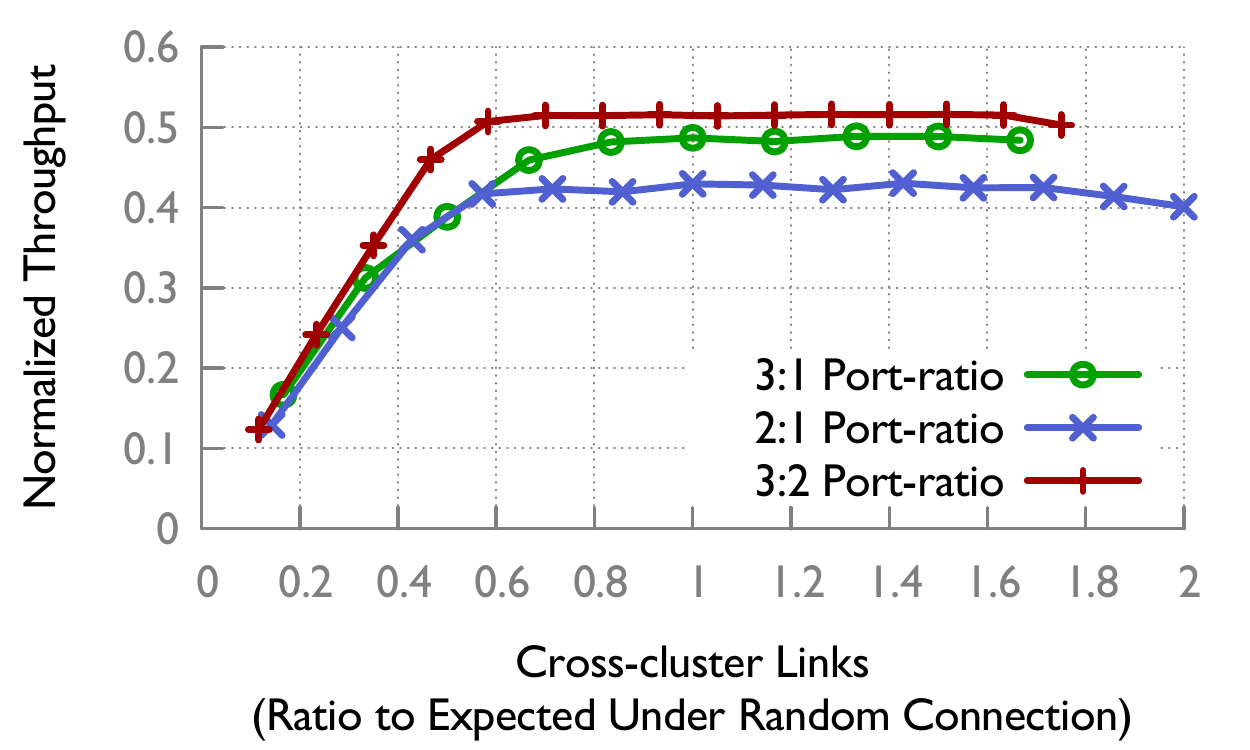}}
\subfigure[]{ \label{fig:hhbias:switchproportion}\includegraphics[width=2.11in]{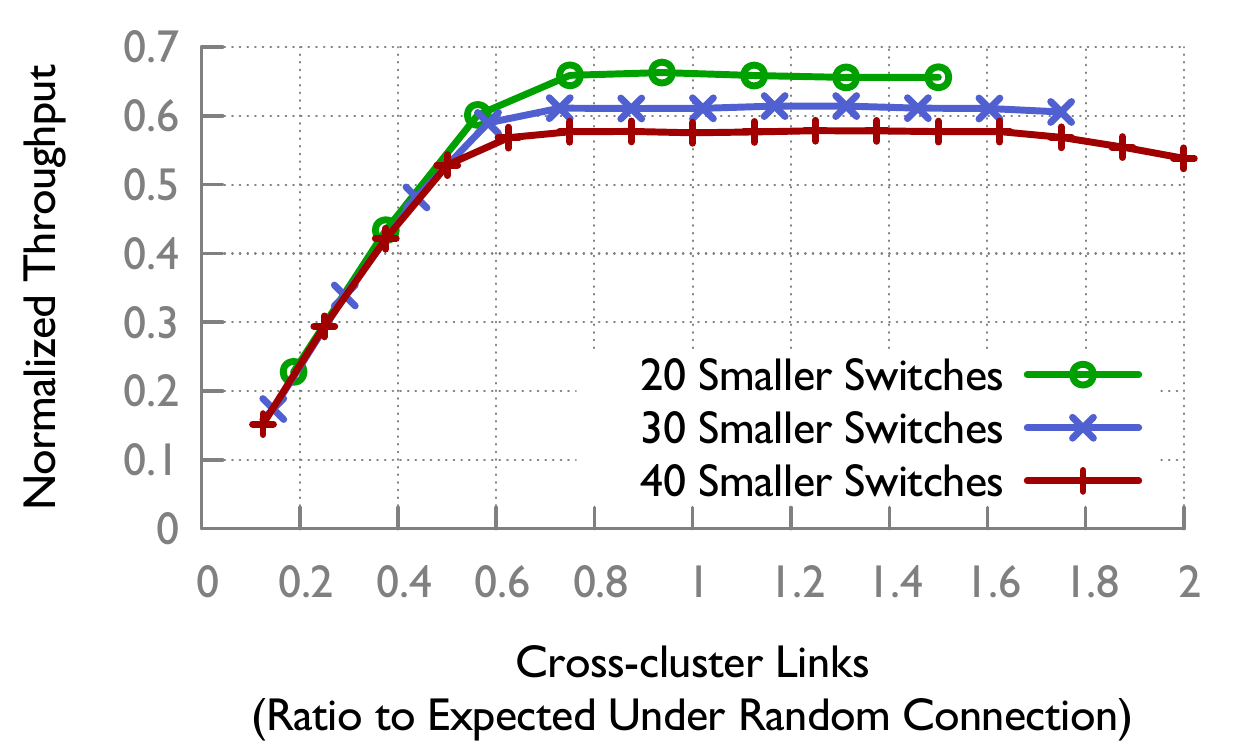}}
\subfigure[]{ \label{fig:hhbias:oversub}\includegraphics[width=2.11in]{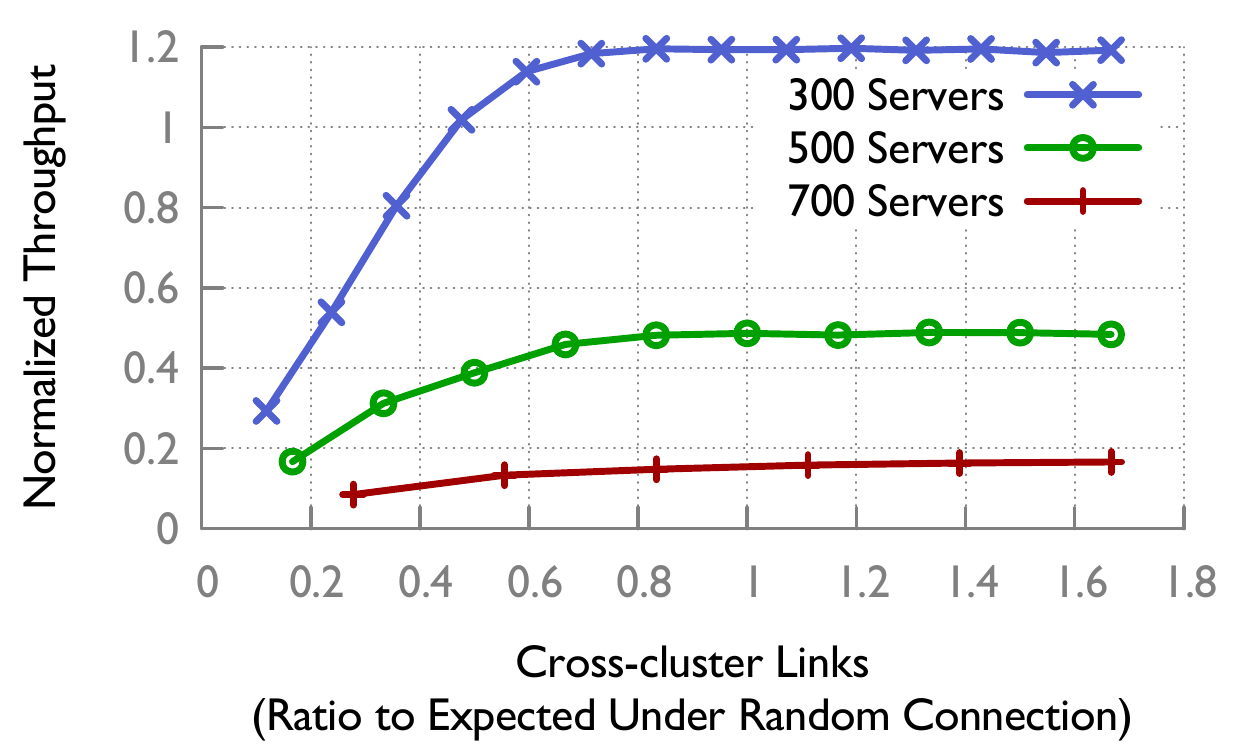}}
\caption{\small \em Interconnecting switches: Peak throughput is stable to a wide range of 
cross-cluster connectivity, regardless of (a) the absolute port counts of switches; (b) the absolute 
counts of switches of each type; and (c) oversubscription in the network.}
\label{fig:bias}
\end{figure*}

\paragraphb{Switch interconnection:} We repeat experiments similar to the
above, but instead of using a uniform random network construction, we vary the 
number of connections across the two clusters of (large and small) 
switches\footnote{Note that specifying connectivity
across the clusters automatically restricts the remaining connectivity to
be within each cluster.}. The distribution of servers is fixed
throughout to be in proportion to the port counts of the switches.

As Fig.~\ref{fig:bias} shows, throughput is surprisingly stable across a
wide range of volumes of cross-cluster connectivity. $x$-axis $=1$ represents the topology
with no bias in construction, \ie vanilla randomness; $x<1$ means the topology is built with
fewer cross-cluster connections than expected with vanilla randomness, etc.
Regardless of the absolute values of the parameters,
when the interconnect has too few connections across the two clusters, throughput drops significantly.
This is perhaps unsurprising -- as our experiments in \S\ref{subsec:cause:exp} will confirm, the cut 
across the two clusters is the limiting factor for throughput in this regime. What \emph{is} surprising, however,
is that across a wide range of cross-cluster connectivity, throughput remains stable at
its peak value. Our theoretical analysis in \S\ref{subsec:analysis} will address this behavior.

\begin{figure}[t]
\centering
\subfigure[]{ \label{fig:bias_dist1}\includegraphics[width=2.3in]{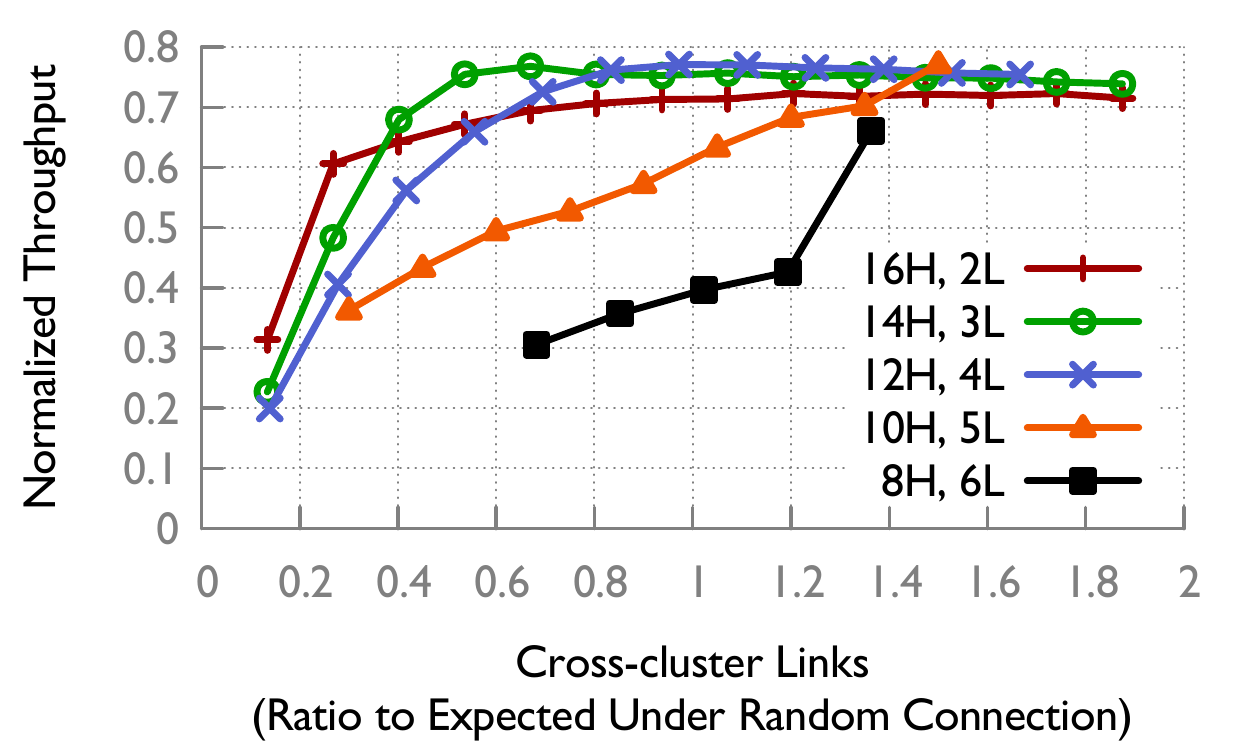}}
\subfigure[]{ \label{fig:bias_dist2}\includegraphics[width=2.3in]{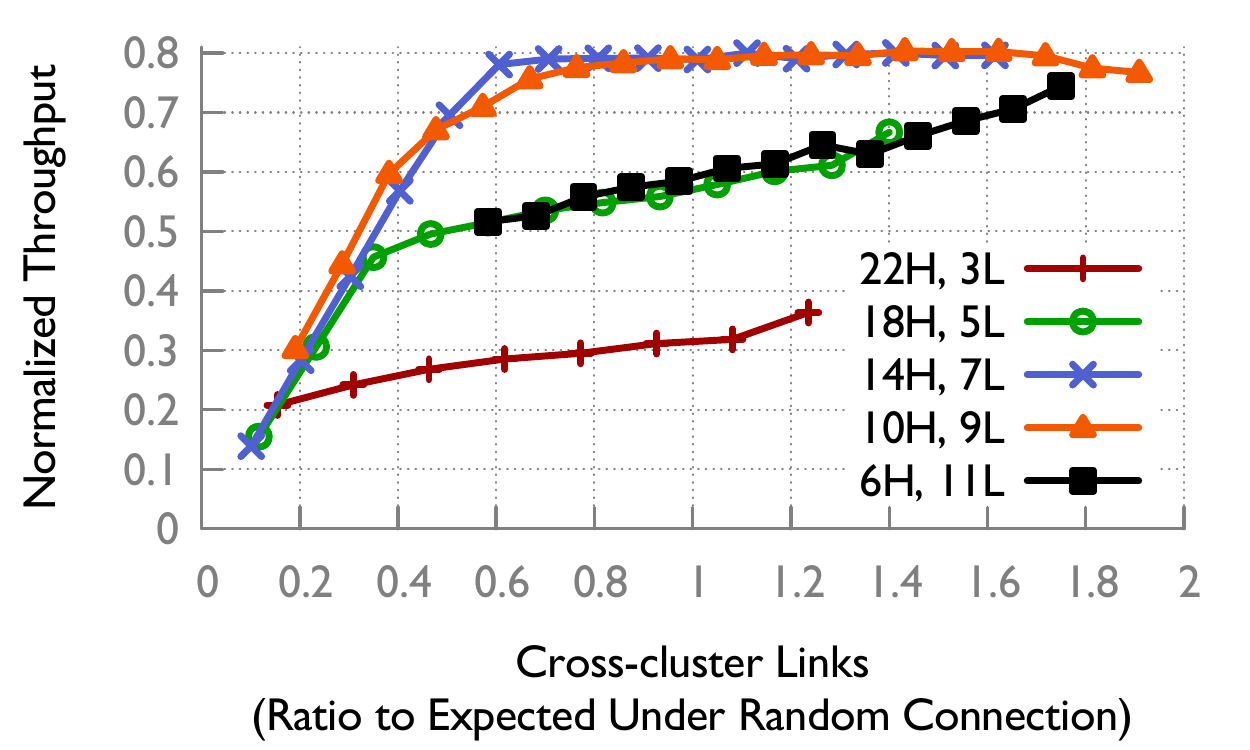}}
\caption{\small \em Combined effect of server distribution and cross-cluster connectivity: Multiple
configurations are optimal, but proportional server distribution with a vanilla random interconnect
is among them. (a) $20$ large, $40$ small switches, with $30$ and $10$ ports respectively. 
(b) $20$ large, $40$ small switches, with $30$ and $20$ ports respectively. Results from $10$ runs.}
\label{fig:bias_dist}
\end{figure}

\paragraphb{Combined effect:} The above results leave open the
possibility that joint optimization across the two parameters (server placement and switch connectivity pattern) can yield better results. Thus, 
we experimented with varying both parameters simultaneously as well. Two representative results
from such experiments are included here. All the data points in Fig.~\ref{fig:bias_dist1} use
the same switching equipment and the same number of servers. Fig.~\ref{fig:bias_dist2}, likewise,
uses a different set of equipment. Each curve in these figures represents a particular 
distribution of servers.  For instance, `$16$H, $2$L' has $16$ servers attached to each larger
switch and $2$ to each of the smaller ones. On the $x$-axis, we again vary the cross-cluster connectivity
(as in Fig.~\ref{fig:hhbias:portproportion}). As the results show, while there are indeed
multiple parameter values which achieve peak throughput, a combination of distributing servers
proportionally (corresponding to `$12$H, $4$L' and `$14$H, $7$L' respectively in the two figures)
and using a vanilla random interconnect is among the optimal solutions. Large deviations
from these parameter settings lead to lower throughput.

%% file: linespeed.tex
\subsection{Heterogeneous Line-speeds}


Data center switches often have ports of different line-speeds, \emph{e.g.}, tens of $1$GbE ports, 
with a few $10$GbE ports. How does this change the above analysis change? 

To answer this question, we modify our scenario such that the small switches still have only low line-speed ports, while the larger
switches have both low line-speed ports and high line-speed ports. \ankitr{The high line-speed ports are assumed to connect only to other high line-speed ports.} We vary both the server distribution and the cross-cluster connectivity
and evaluate these configurations for throughput. As the results in Fig.~\ref{fig:linespeed1} indicate, the picture is not as clear
as before, with multiple configurations having nearly the same throughput. Each curve corresponds to one particular distribution of
servers across switches. For instance, `$36$H, $7$L' has $36$ servers attached to each large switch, and $7$ servers attached to
each small switch. The total number of servers across all curves is constant. While we are unable to make clear qualitative claims of the
nature we made for scenarios with uniform line-speed, our simulation tool can be used to determine the optimal configuration for such scenarios.

We also investigate the impact of the number and the line-speed of the high line-speed ports on the large switches. For these tests, we
fix the server distribution, and vary cross-cluster connectivity. We measure throughput for various `high' line-speeds 
(Fig.~\ref{fig:linespeed2}) and numbers of high line-speed links (Fig.~\ref{fig:linespeed3}). While higher number or 
line-speed does increase throughput, its impact diminishes when cross-cluster connectivity is too small. 
This is expected: as the bottlenecks move to the cross-cluster edges, having high capacity between the 
large switches does not increase the \emph{minimum} flow.

\begin{figure*}
\centering
\subfigure[]{ \label{fig:linespeed1}\includegraphics[width=2.11in]{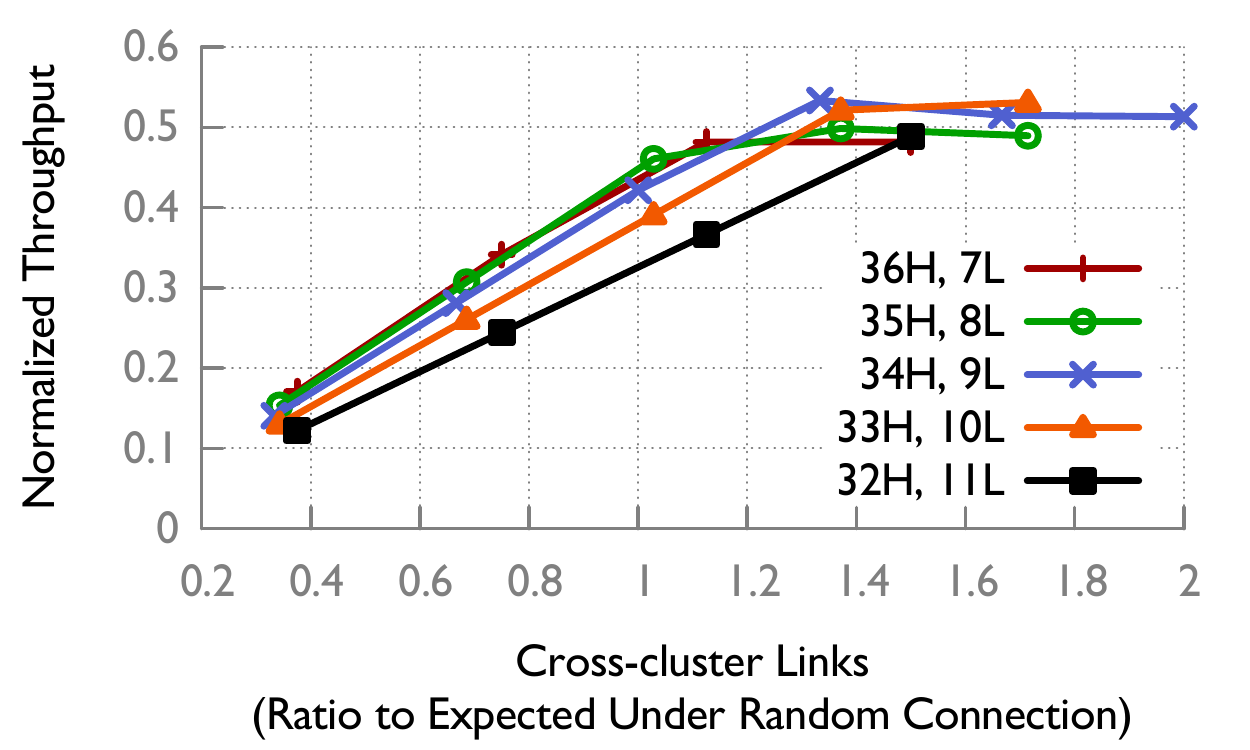}}
\subfigure[]{ \label{fig:linespeed2}\includegraphics[width=2.11in]{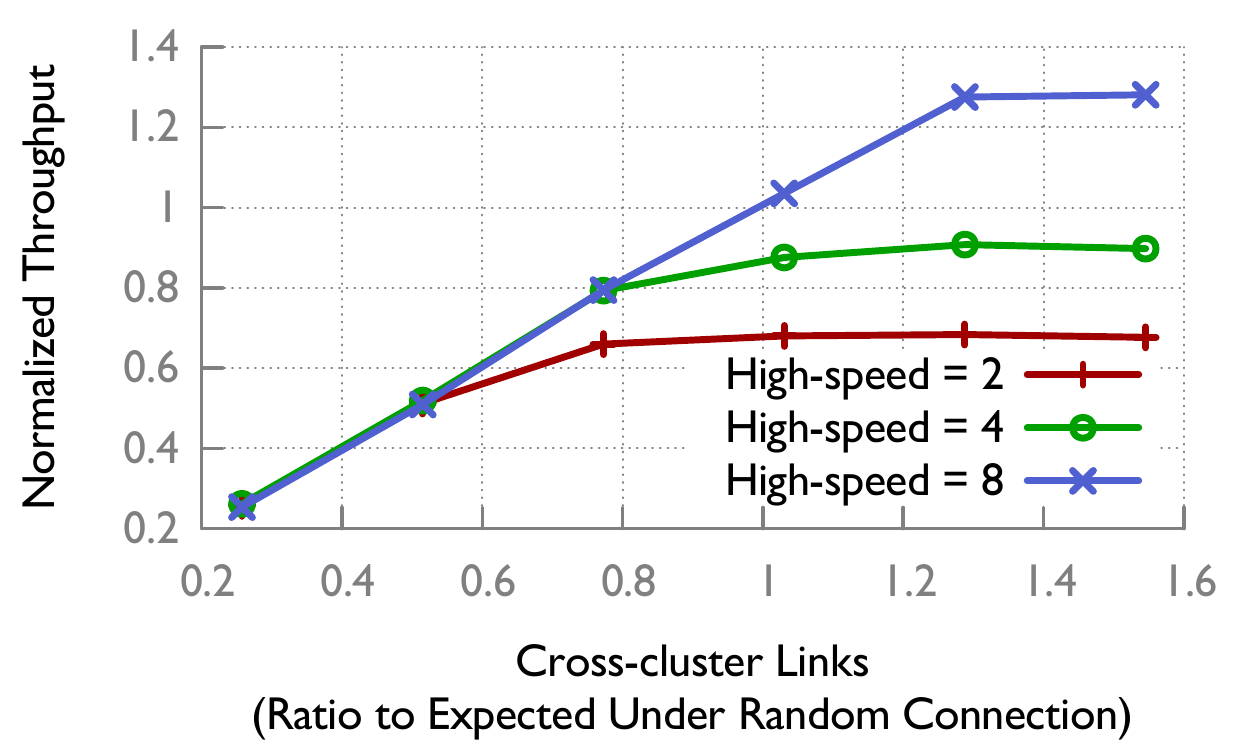}}
\subfigure[]{ \label{fig:linespeed3}\includegraphics[width=2.11in]{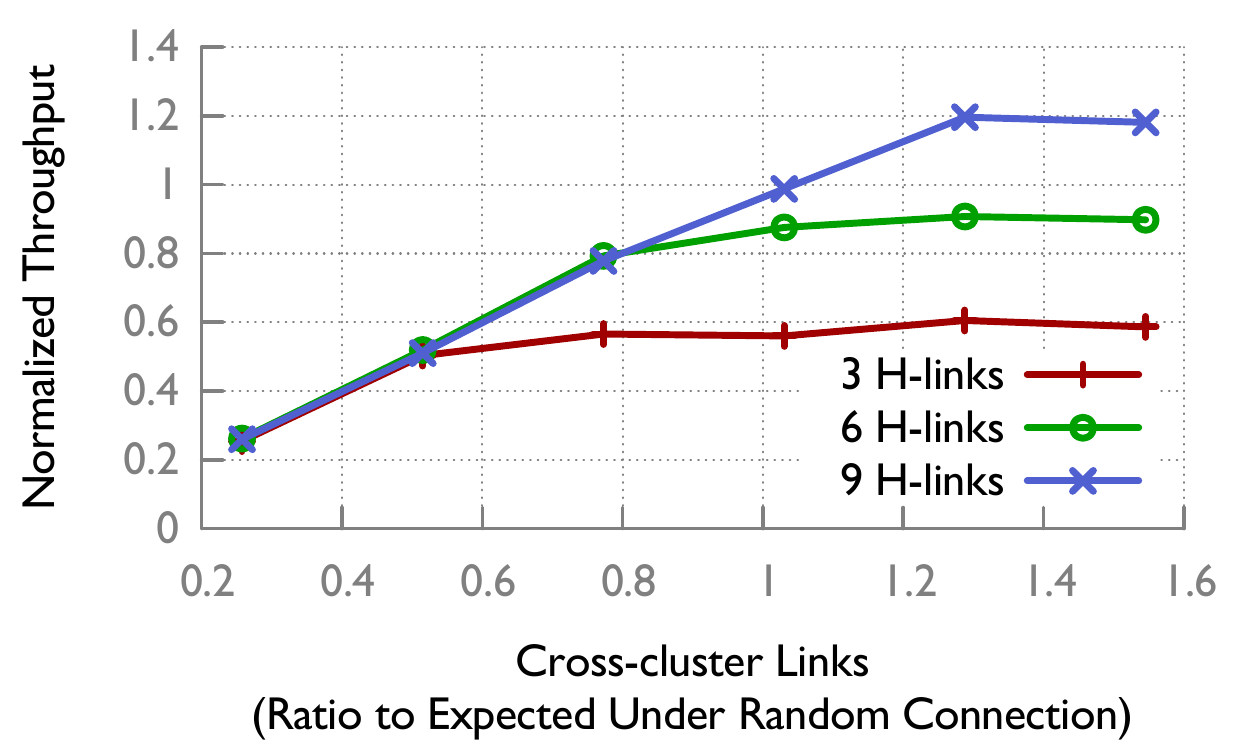}}
\caption{\small \em Throughput variations with the amount of cross-cluster connectivity:
(a) various server distributions for a network with $20$ large and $20$ small switches, 
with $40$ and $15$ low line-speed ports respectively, with the large switches having $3$ additional $10\times$ capacity connections; (b) with 
different line-speeds for the high-speed links keeping their count fixed at $6$ per large switch; and (c) with 
different numbers of the high-speed links at the big switches, keeping their line-speed fixed at $4$ units.}
\label{fig:linespeed}
\end{figure*}

In the following, we attempt to add more than just the intuition for our results. We seek to explain
throughput behavior by analyzing factors such as bottlenecks, total network utilization, shortest path lengths between
nodes, and the path lengths actually used by the network flows.

%% file: cause.tex
\section{Explaining Throughput Results}
\label{sec:cause}

We investigate the cause of several of the throughput effects we observed in the previous section.  First, in \S\ref{subsec:cause:exp}, we break down throughput into component factors --- network utilization, shortest path length, and ``stretch'' in paths --- and show that the majority of the throughput changes are a result of changes in utilization, though for the case of varying server placement, path lengths are a contributing factor. Note that a decrease in utilization corresponds to a saturated bottleneck in the network.

Second, in \S\ref{subsec:analysis}, we explain in detail the surprisingly stable throughput observed over a wide range of amounts of connectivity between low- and high-degree switches.  We give an upper bound on throughput, show that it is empirically quite accurate in the case of uniform line-speeds, and give a lower bound that matches within a constant factor for a restricted class of graphs.  We show that throughput in this setting is well-described by two regimes: (1) one where throughput is limited by a sparse cut, and (2) a ``plateau'' where throughput depends on two topological properties: total volume of connectivity and average path length $\apl$.  The transition between the regimes occurs when the sparsest cut has a fraction $\Theta(1/\apl)$ of the network's total connectivity.

Note that bisection bandwidth, a commonly-used measure of network capacity which is equivalent to the sparsest cut in this case, begins falling as soon as the cut between two equal-sized groups of switches has less than $\frac{1}{2}$ the network connectivity.  Thus, our results demonstrate (among other things) that bisection bandwidth is not a good measure of performance\footnote{This result is explored further in followup work~\cite{sigmetricsdraft}, where we point out problems with bisection bandwidth as a performance metric.}, since it begins falling asymptotically far away from the true point at which throughput begins to drop.

\subsection{Experiments}
\label{subsec:cause:exp}

Throughput can be exactly decomposed as the product of four factors: \[ T = \frac{C \cdot U}{\apl \cdot AS} = C \cdot U \cdot \frac{1}{\apl} \cdot \frac{1}{AS} \] where $C$ is the total network capacity, $U$ is the average link utilization, $\apl$ is the average shortest path length, and $AS$ is the average stretch, i.e., the ratio between average length of routed flow paths\footnote{This average is weighted by amount of flow along each route.} and $\apl$.  Throughput may change due to any one of these factors.  For example, even if utilization is $100\%$, throughput could improve if rewiring links reduces path length (this explained the random graph's improvement over the fat-tree in~\cite{jellyfish}).  On the other hand, even with very low $\apl$, utilization and therefore throughput will fall if there is a bottleneck in the network.

We investigate how each of these factors influences throughput (excluding $C$ which is fixed).
Fig.~\ref{fig:cause} shows throughput ($T$), utilization ($U$), inverse shortest path length ($1/\apl$), and inverse stretch ($1/AS$).  An increase in any of these quantities increases throughput.  To ease visualization, for each metric, we normalize its value with respect to its value when the throughput is highest so that quantities are unitless and easy to compare.

Across experiments, our results (Fig.~\ref{fig:cause}) show that high utilization best explains high throughput. 
Fig.~\ref{fig:cause:svrdist_oversub} analyzes the throughput results for `$480$ Servers' from 
Fig.~\ref{fig:serverdist:oversub}, Fig.~\ref{fig:cause:bias_oversub} corresponds to `$500$
Servers' in Fig.~\ref{fig:hhbias:oversub}, and Fig.~\ref{fig:cause:linespeed} to `$3$ H-links' in
Fig.~\ref{fig:linespeed3}. Note that it is not 
obvious that this should be the case: Network utilization would also be high if the flows took
long paths and used capacity wastefully. At the same time, one could reasonably expect `Inverse Stretch'
to also correlate with throughput well --- if the paths used are close to shortest, then the flows are not
wasting capacity.  Path lengths do play a role --- for example, the right end of Fig.~\ref{fig:cause:svrdist_oversub} shows an increase in path lengths, explaining why throughput falls about $25\%$ more than utilization falls --- but the role is less prominent than utilization.

Given the above result on utilization, we examined where in the network the corresponding bottlenecks occur. From our linear
program solver, we are able to obtain the link utilization for each network link. We averaged link
utilization for each link type in a given network and flow scenario \ie computing average
utilization across links between small and large switches, links between small switches only, etc.
The movement of under-utilized links and bottlenecks shows clear correspondence to our throughput
results. For instance, for Fig.~\ref{fig:hhbias:oversub}, as we move leftward along the $x$-axis,
the number of links across the clusters decreases, and we can expect bottlenecks to manifest at these
links. This is exactly what the results show. For example, for the leftmost point ($x=1.67$, $y=1.67$)
on the `$500$ Servers' curve in Fig.~\ref{fig:hhbias:oversub},
links inside the large switch cluster are on average $<20\%$ utilized while the links between across clusters are
close to fully utilized ($>90\%$ on average). On the other hand, for the points with higher throughput,
like ($x=1$, $y=0.49$), all network links show uniformly high utilization ($\sim$$100\%$). Similar
observations hold across all our experiments.

\begin{figure*}
\centering
\subfigure[]{ \label{fig:cause:svrdist_oversub}\includegraphics[width=2.11in]{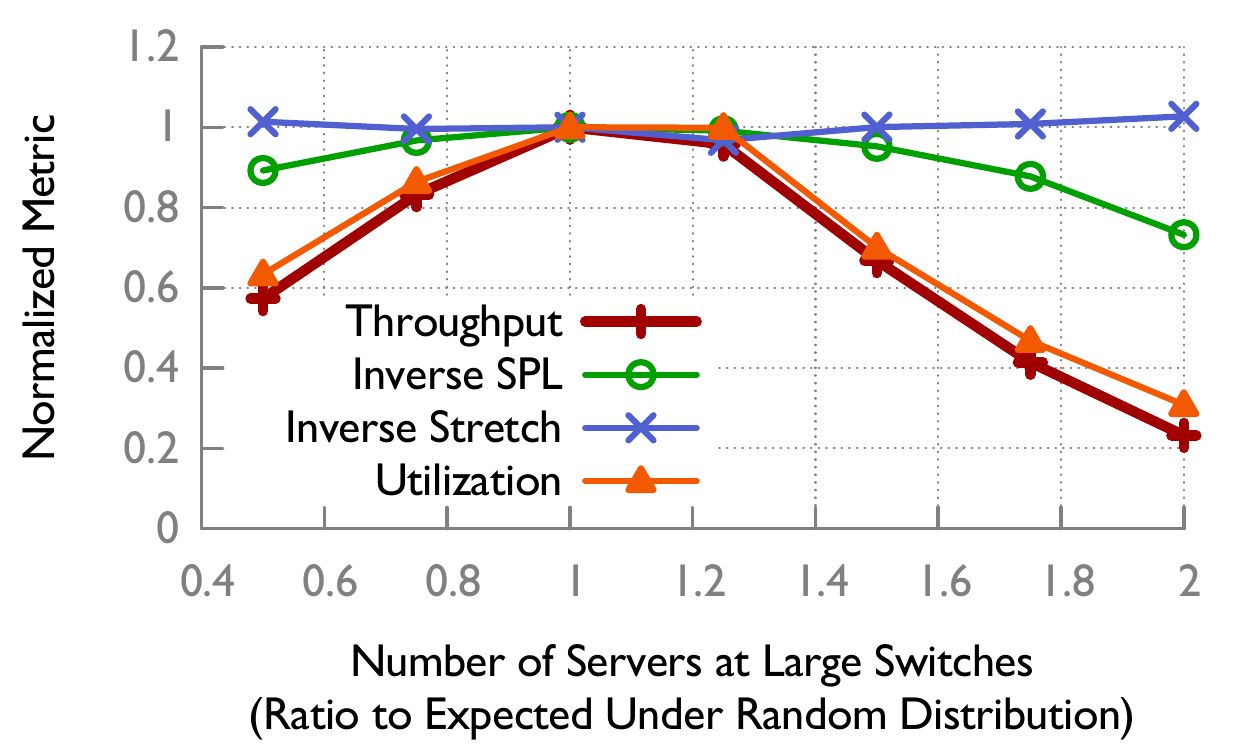}}
\subfigure[]{ \label{fig:cause:bias_oversub}\includegraphics[width=2.11in]{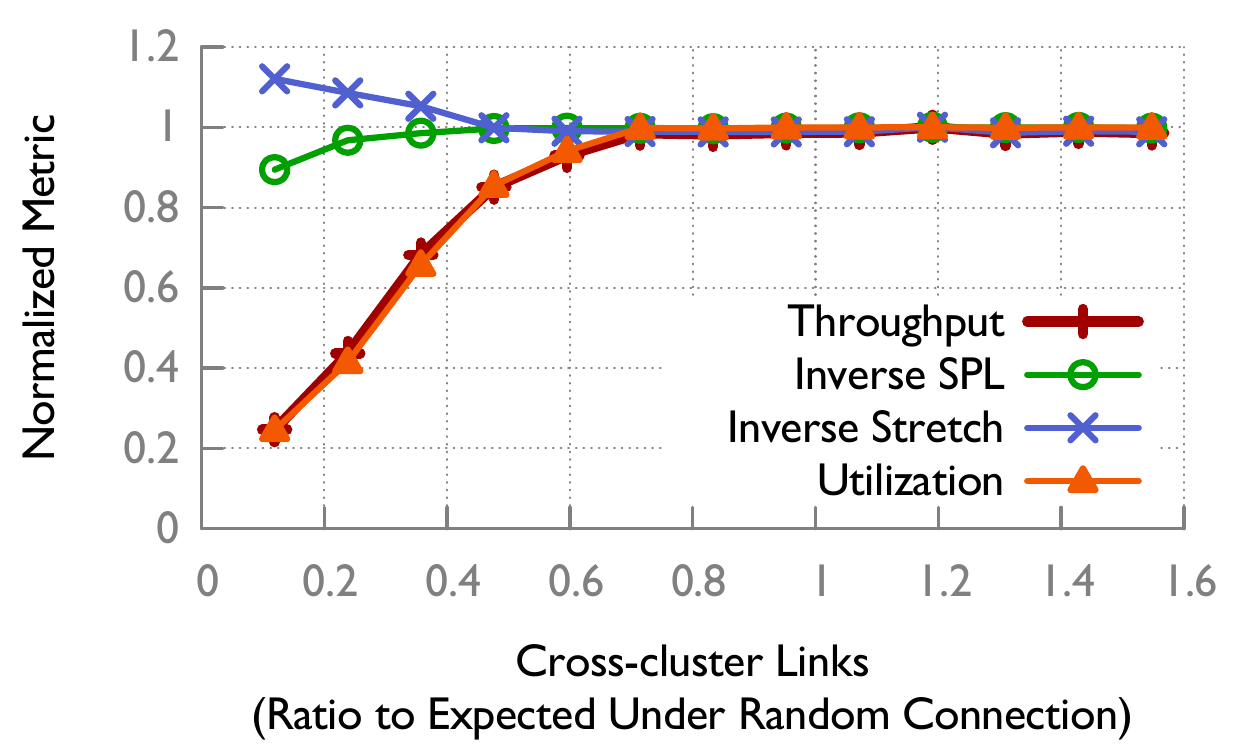}}
\subfigure[]{ \label{fig:cause:linespeed}\includegraphics[width=2.11in]{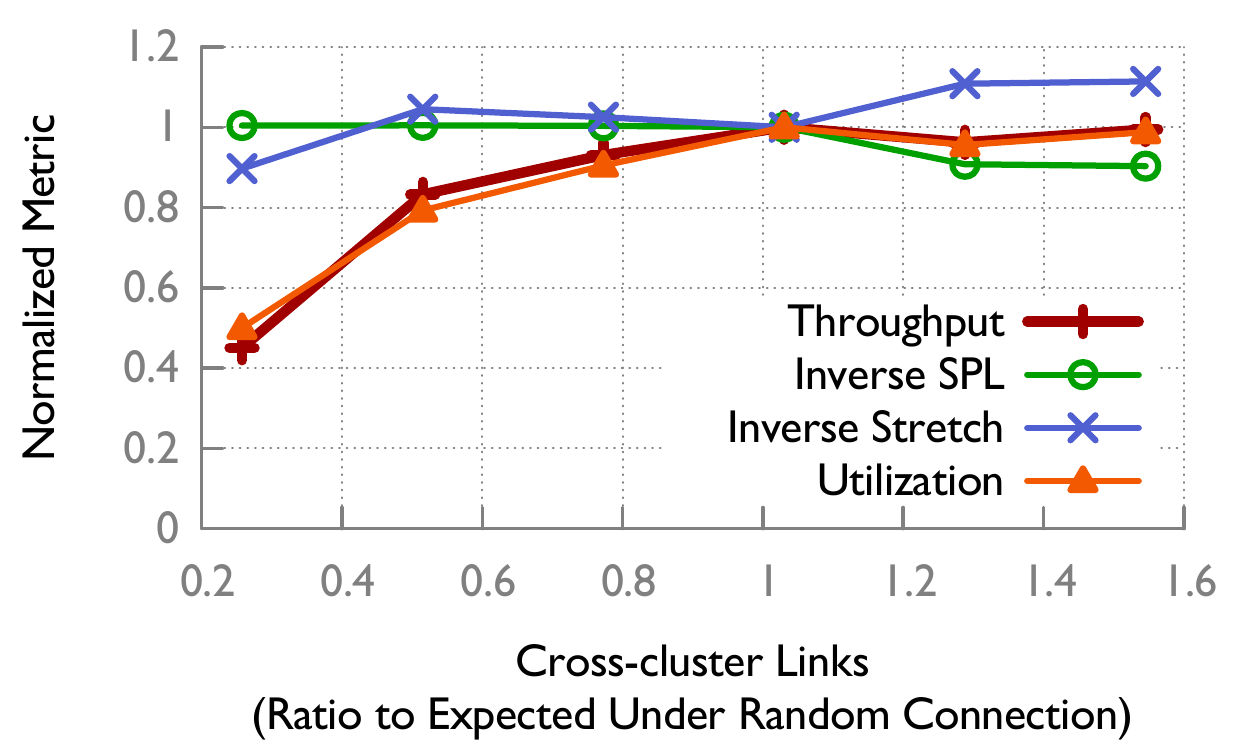}}
\caption{\small \em The dependence of throughput on all three relevant factors: inverse path length, inverse stretch,
and utilization. Across experiments, total utilization best explains throughput, indicating that bottlenecks govern
throughput.}
\label{fig:cause}
\end{figure*}

%% file: analysis_new.tex
\renewcommand{\apl}{\langle D \rangle}

\subsection{Analysis}
\label{subsec:analysis}

Fig.~\ref{fig:bias} shows a surprising result: network throughput is stable at its peak
value for a wide range of cross-cluster connectivity. In this section, we provide upper and lower bounds on throughput to explain the result.  Our upper bound is empirically quite close to the observed throughput in the case of networks with uniform line-speed.  Our lower bound applies to a simplified network model and matches the upper bound within a constant factor.  This analysis allows us to identify the point (\ie amount of cross-cluster connectivity) where throughput begins to drop, so that our topologies can avoid this regime, while allowing flexibility in the interconnect.


\begin{figure}[t]
\centering
\subfigure[]{ \label{fig:bounds:bound_hhbias}\includegraphics[width=2.25in]{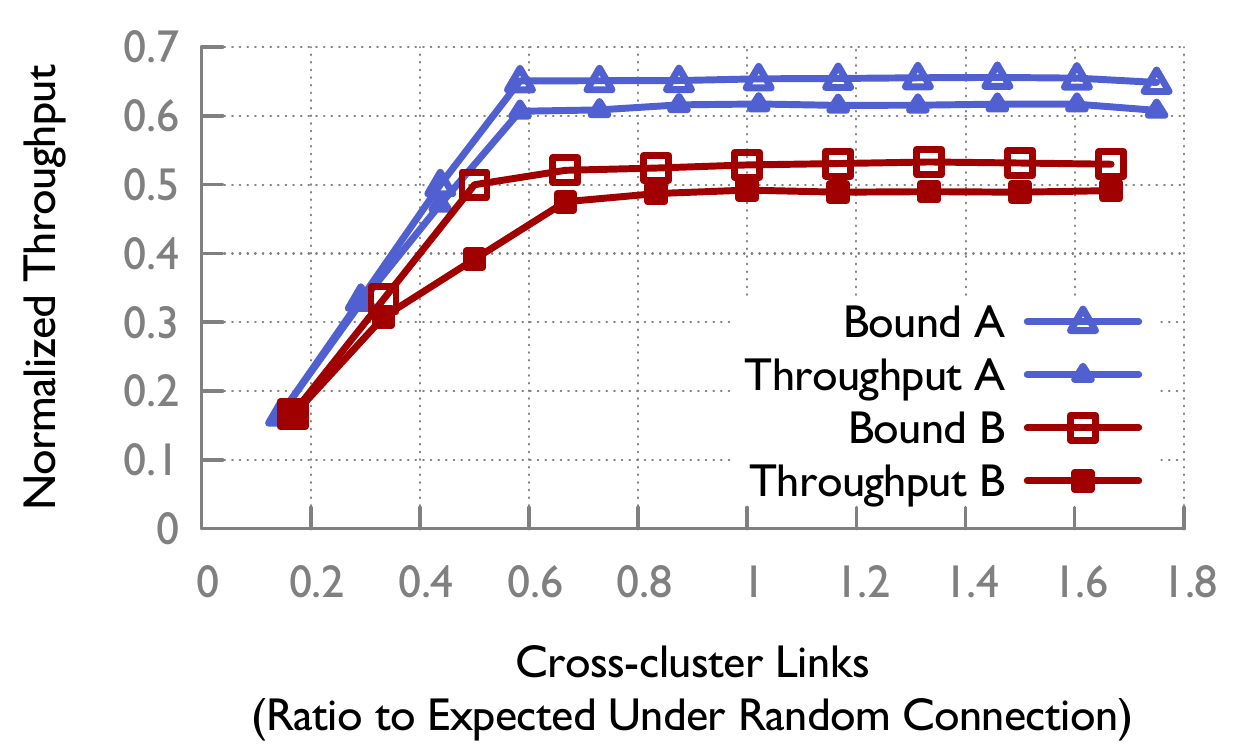}}
\subfigure[]{ \label{fig:bounds:bound_linespeed}\includegraphics[width=2.25in]{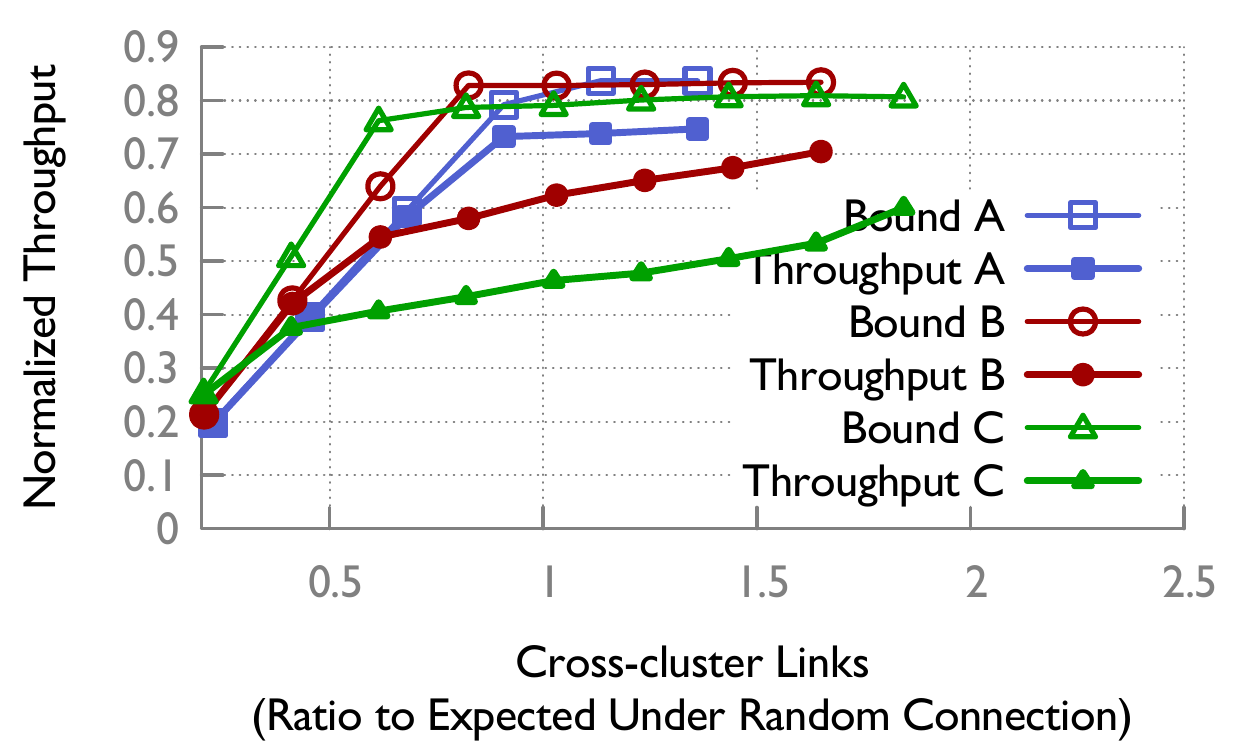}}
\caption{\small \em Our analytical throughput bound is close to the observed throughput for the uniform line-speed scenario (a) for which the bound and the corresponding throughput are shown for two representative cases $A$ and $B$, but can be quite loose with non-uniform line-speeds (b).}
\label{fig:bounds}
\end{figure}

\paragraphb{Upper-bounding throughput.}  We will assume the network is composed of two ``clusters'', which are simply arbitrary sets of switches, with $n_1$ and $n_2$ attached servers respectively.  Let $C$ be the sum of the capacities of all links in the network (counting each direction separately), and let $\bar{C}$ be that of the links crossing the clusters. To simplify this exposition, we will assume the number of flows crossing between clusters is exactly the expected number for random permutation traffic: $n_1 \frac{n_2}{n_1 + n_2} + n_2 \frac{n_1}{n_1 + n_2} = \frac{2 n_1 n_2}{n_1 + n_2}$.  Without this assumption, the bounds hold for random permutation traffic with an asymptotically insignificant additive error.

Our upper bound has two components.  First, recall our path-length-based bound from \S\ref{sec:jellyfish} shows the throughput of the minimal-throughput flow is $T \leq \frac{C}{\apl f}$ where $\apl$
is the average shortest path length and $f$ is the number of flows. For random permutation traffic, $f=n_1 + n_2$.

Second, we employ a cut-based bound. The cross-cluster flow is $\geq T \frac{2 n_1 n_2}{n_1 + n_2}$. This flow is bounded above by the
capacity $\bar{C}$ of the cut that separates the clusters, so we must have $T \leq \bar{C}\frac{n_1 + n_2}{2 n_1 n_2}$.

Combining the above two upper bounds, we have
\begin{equation} \label{eqn:throughput-upper}
	T \leq min \left\{ \frac{C}{\apl (n_1 + n_2)}, \frac{\bar{C}(n_1 + n_2)}{2 n_1 n_2}\right\}
\end{equation}
Fig.~\ref{fig:bounds} compares this bound to the actual observed throughput for two
cases with uniform line-speed (Fig.~\ref{fig:bounds:bound_hhbias}) and
a few cases with mixed line-speeds (Fig.~\ref{fig:bounds:bound_linespeed}). The bound
is quite close for the uniform line-speed setting, both for the cases presented here and several other experiments we conducted, but can be looser for mixed line-speeds.

The above throughput bound begins to drop when the cut-bound begins to dominate.  In the special case that the two clusters have equal size, this point occurs when
\begin{equation}
	\label{eqn:drop-point-upper}
	\bar{C} \leq \frac{C}{2\apl}.
\end{equation}

A drop in throughput when the cut capacity is inversely proportional to average shortest path length  has an intuitive explanation.  In a random graph, most flows have many shortest or nearly-shortest paths.  Some flows might cross the cluster boundary once, others might cross back and forth many times.  In a uniform-random graph with large $\bar{C}$, near-optimal flow routing is possible with any of these route choices.  As $\bar{C}$ diminishes, this flexibility means we can place some restriction on the choice of routes without impacting the flow.  However, the flows which cross clusters must still utilize at least one cross-cluster hop, which is on average a fraction $1/\apl$ of their hops.  Therefore in expectation, since $\frac{1}{2}$ of all (random-permutation) flows cross clusters, at least a fraction $\frac{1}{2\apl}$ of the total traffic volume will be cross-cluster.  We should therefore expect throughout to diminish once less than this fraction of the total capacity is available across the cut, which recovers the bound of Equation~\ref{eqn:drop-point-upper}.

However, while Equation~\ref{eqn:drop-point-upper} determines when the \emph{upper bound} on throughput drops, it does not does not bound the point at which \emph{observed} throughput drops: since the upper bounds might not be tight, throughput could drop earlier or later.  We can, however, construct a bound based on a given throughput value.  Suppose that the optimal throughput, in any configuration, is $T^*$.  Then since $T^* \leq \bar{C}\frac{n_1 + n_2}{2 n_1 n_2}$, throughput must drop below $T^*$ when $\bar{C}$ is less than $C^* := T^* \frac{2 n_1 n_2}{n_1 + n_2}$.  If we are able to empirically estimate $T^*$ (which is not
unreasonable, given its stability), we can determine the value of $\bar{C}^*$ below which throughput \emph{must} drop.

In Fig.~\ref{fig:cval}, we test $18$ different configurations using two
clusters with increasing cross-cluster connectivity (equivalently, $\bar{C}$). 
The one point marked on each curve corresponds to the $\bar{C}^*$ threshold calculated above. 
As predicted, below $\bar{C}^*$, throughput is less than its peak value.

\paragraphb{Lower-bounding throughput.} Here we lower-bound throughput in a restricted class of random graphs.  We show that our throughput upper bound (Eqn.~\ref{eqn:throughput-upper}), and the drop point of Eqn.~\ref{eqn:drop-point-upper}, are correct within constant factors in this case.

\begin{figure}
\centering
\includegraphics[width=2.6in]{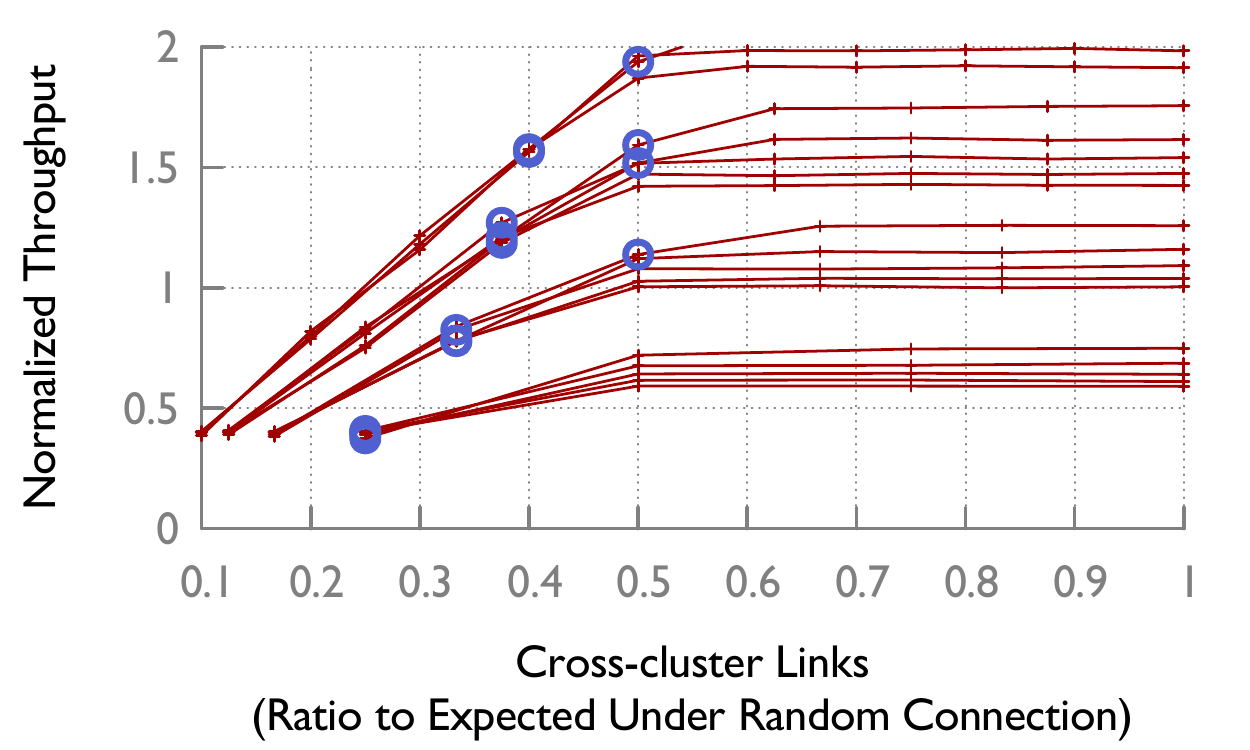}
\caption{\small \em Throughput shows a characteristics profile with respect to varying levels of cross-cluster connectivity. The one
point marked on each curve indicates our analyticallly determined threshold of cross-cluster connectivity below which
throughput must be smaller than its peak value.}
\label{fig:cval}
\end{figure}




We restrict this analysis to networks $G = (V,E)$ in which all $n$ nodes have constant degree $d$, all links have capacity $1$ in each direction, and the vertices $V$ are grouped into two equal size clusters $V_1, V_2$, i.e., $|V_1| = |V_2| = \frac{1}{2} n$.  Let $p,n$ be such that each node has $pn$ neighbors within its cluster and $qn$ neighbors in the other cluster, so that $p+q=d/n = \Theta(1/n)$.  Under this constraint, we choose the remaining graph from the uniform distribution on all $d$-regular graphs. Thus, for each of the graphs under consideration, the total inter-cluster connectivity is $\bar{C} = 2q\cdot |V_1|\cdot|V_2| =q\cdot \frac{n^2}{2}$. Decreasing $q$ corresponds to decreasing the cross-cluster connectivity and increasing the connectivity within each cluster. All our results below hold with high probability (w.h.p.) over the random choice of the graph.  Let $T(q)$ be the throughput with the given value of $q$, and let $T^*$ be the throughput when $p=q$ (which will also be the maximum throughput).

Our main result is the following theorem, which explains the throughput results by proving that while $q\geq q^*$, for some value $q^*$ that we determine, the throughput $T(q)$ is within a constant factor of $T^*$. Further, when $q<q^*$, $T(q)$ decreases roughly linearly with $q$.

\begin{theorem}\label{thm:throughput}
There exist constants $c_1,c_2$ such that if $q^*=c_1\frac{1}{\langle D \rangle}p$, then for $q\geq q^*$ w.h.p. $T(q) \geq c_2T^*$. For $q<q^*$, $T(q) = \Theta(q)$.
\end{theorem}

Our proof consists of four parts. First, in Lemma~\ref{lemma:peakthput}, we compute the peak value of throughput (within constant factors) $T^*$. In Lemma~\ref{lemma:crossIsSparsest}, we show that the sparsest cut value (defined below) is linear in $q$ for a bipartite demand graph across the clusters\footnote{In general, the sparsest cut is NP-Hard to compute. It is the specific setting that makes this possible.}. In Lemma~\ref{lemma:mainlemma}, we show that for $q\leq q^*$, throughput is within a constant factor of the sparsest cut value and thus reduces linearly with $q$. Finally, we show that for $q > q^*$, throughput is within a constant factor of its peak value.

We will use a celebrated result that can be found in \cite{LLR95} as Theorem $4.1$. We paraphrase it here to suit our needs:

\begin{theorem}[Linial, London, Rabinovich]\label{thm:LLR}
We are given a network $G=(V,E,C)$ with vertices $V$, edges $E$, and their capacities $C$. We are also given a demand graph $H=(V, E')$ with $k=|E'|$ source-sink pairs. For a set $S\subseteq V$, let $\mathsf{Cap}(S)$ be the sum of the capacities of edges connecting $S$ and $S'$ and $\mathsf{Dem}(S)$ be the number of source-sink pairs separated by $S$. Let $T(G,H)$ be the throughput for given $G$ and $H$. Then there exists a set $S\subseteq V$ such that
$$\frac{\mathsf{Cap}(S)}{\mathsf{Dem}(S)}\leq O(\log k)\cdot T(G,H)$$
\end{theorem}

The minimum of the ratios $\frac{\mathsf{Cap}(S)}{\mathsf{Dem}(S)}$, \ie $\min_{S\subseteq V}\frac{|E_G(S,S')|}{|E_H(S,S')|}$ is referred to as the non-uniform sparsest cut of graph $G$ with a demand graph $H$~\cite{LLR95}. Then, the above theorem immediately implies the following relationship between the sparsest cut $\phi$ and throughput $T$ for a graph $G$ and demand graph $H$:

\begin{equation}\label{eq:sc_vs_tp}
\phi(G,H)\leq O(\log k)\cdot T(G,H)
\end{equation}

In the below, $K_{V_1,V_2}$ refers to the \emph{complete bipartite} demand graph where each node communicates with (and only with) all nodes in the opposite cluster. We shall use this demand graph to prove our results, and then show later that throughput under this demand graph is within a constant factor of throughput under random permutations. 

\begin{lemma}\label{lemma:peakthput}
When $p=q=q_0$, for demand graph $H=K_{V_1,V_2}$, $T(q=q_0)=T^*=\Theta(\frac{1}{n \log n})$.
\end{lemma}
\begin{proof}

For $q=p$, it is well known~\cite{ellis} that $G$ is an almost optimal expander with high probability, and all the balanced cuts have about the same number of edges being cut, which is $O(d\cdot n)$. Thus the (non-uniform) sparsest cut value is:
\begin{equation}\label{eq:basecase}
\phi(q_0)=\Theta\left(\frac{d\cdot n}{n^2}\right)=\Theta\left(\frac{d}{n}\right)
\end{equation}

From equation \ref{eq:sc_vs_tp}, we obtain that for some constant $c$,
\begin{equation}\label{eq:flowzero}
T(q_0) \geq c\frac{1}{\log k}\phi(q_0)\geq \Omega\left(\frac{1}{\log n}\right)\left(\frac{d}{n}\right)
\end{equation}

For constant $d$, we obtain:
\begin{equation}\label{eq:t_omega}
 T(q_0) \geq \Omega\left(\frac{1}{n\log n}\right)
 \end{equation}

Next, we invoke our path-length based bound: $T \leq \frac{|E|}{\langle D \rangle f}$, which, in this setting implies $T \leq O(\frac{n d}{\apl n^2})$.
Under our graph model (and trivially for $d$-regular graphs), the following result holds \cite{chunglu1,chunglu2} for average shortest path length $\apl$:
 \begin{equation}\label{eq:APL}
 \apl \geq \Omega\left(\frac{\log n}{\log d}\right)
 \end{equation}

Using this result, for constant $d$, we obtain $T \leq O(\frac{1}{n\log n})$, which, together with equation~\ref{eq:t_omega}, yields the lemma's result.
\end{proof}

\begin{lemma}\label{lemma:crossIsSparsest}
For $H=K_{V_1, V_2}$, $\phi(G,H) = \Theta(q)$.
\end{lemma}
\begin{proof}
In the most general case, a cut in G can be described by the vertex sets $S=(k_1 \in V_1) \cup (k_2 \in V_2)$ and $S'=V \setminus S$, so that abritrary subsets $k_1$ and $k_2$ of $V_1$ and $V_2$ respectively, are separated from the rest of the graph by the cut. Then:
 \begin{multline}\label{eq:edgesum}
 E_G(S,S') = E_G(k_1, V_1 \setminus k_1) + E_G(k_2, V_2 \setminus k_2) + \\ E_G(k_1, V_2 \setminus k_2) + E_G(k_2, V_1 \setminus k_1)
\end{multline}

Note that $k_1$ and $k_2$ are both subgraphs of random regular graphs $V_1$ and $V_2$ of degree $pn$ (using only the internal edges of each cluster). Also, across the clusters $V_1$ and $V_2$, we have a bipartite expander graph of degree $qn$. According to the expander mixing lemma~\cite{ellisMixing}, the number of cut-edges across subgraphs of each of these expanders is within a constant factor of the expected number of edges.
Thus, for some constants $c_l$, $c_m$, and $c_n$:
\begin{equation}\label{eq:expander1}
 E_G(k_1, V_1 \setminus k_1) \geq c_l pn  k_1 \frac {n/2 - k_1}{n/2} = c_l k_1 pn (1 - 2k_1/n)
 \end{equation}

\begin{equation}\label{eq:expander2}
 E_G(k_2, V_2 \setminus k_2) \geq c_m k_2 pn (1 - 2k_2/n) 
 \end{equation}

 \begin{multline}\label{eq:expander3}
E_G(k_1, V_2 \setminus k_2) + E_G(k_2, V_1 \setminus k_1) \geq \\ c_n (k_1 qn (1 - 2k_2/n) + k_2 qn (1 - 2k_1/n))
\end{multline}

Using $c_{min} = min\{c_l, c_m, c_n\}$, $k = k_1 + k_2$, and degree $d = pn + qn$, we obtain (after simplification) from the above equations: 

\begin{equation}
E_G(S,S') \geq c_{min} (kd + 2pk^2 + 4k_1k_2d/n)
\end{equation}

For $H=K_{V_1, V_2}$, $E_H(S,S') = kn/2 - 2k_1k_2$. With a fixed $k$, it is easy to show that $E_G(S,S')/E_H(S,S')$ is minimized when $(k_1,k_2)=(0,k)$ or $(k,0)$; the minimum value being $c_{min} \frac{2d - 4kp}{n}$. For $k \in (0,n/2]$, the minimum value of this expression is $c_{min}\frac{2d - 2pn}{n} = 2qc_{min}$. Thus $E_G(S,S')/ E_H(S,S') \geq 2q c_{min}$, and further, $\phi(G, H)  = \min_{S\subseteq V}\frac{|E_G(S,S')|}{|E_H(S,S')|} \geq 2q c_{min}$. To conclude the lemma's proof, we note that $\frac{|E_G(V_1,V_2)|}{|E_H(V_1,V_2)|} = 2q$ implies that $\phi(G, H) \leq 2q$.
\end{proof}

\begin{lemma}\label{lemma:mainlemma}
For any constant $c_1$, if $q^*=c_1\frac{1}{\apl}p$, then for $q<q^*$, $T(G,H) \leq \phi(G,H) = \Theta(q)$. Further, there is a constant $c_2$ (that depends on $c_1$) such that $T(G,H)\geq c_2\phi(G,H) = c_2q$. Thus, for $q<q^*$, $T(G,H) = \Theta(q)$.
\end{lemma}

\begin{proof}

In Lemma~\ref{lemma:crossIsSparsest}, we have shown that $\phi(G,H) = \phi(q)=\Theta(q)$. This allows us to conclude that $T(G,H) \leq \phi(G,H) = \Theta(q)$, since the flow cannot be greater than the sparsest cut.

To show that $T(G,H)\geq c_2\phi(G,H)$ for $q < q^*$, it suffices to show that a flow of value $\Theta(q)$ can be supported on our network. In the following, we show the existence of such a flow, sending $\Theta(q)$ units between every pair of nodes $(u,v) \in V_1 \times V_2$. 

With each node $u \in V_1$ having $qn$ edges to $V_2$, we have $qn^2$ edges across the $V_1$-$V_2$ cut. We route flow $\Theta(q)$ between each $(u,v)$ as follows: $u$ splits the flow equally to all nodes in $V_1$ sending each $\Theta(q/n)$. Each receiving node $l \in V_1$ further splits the flow equally across all its $qn$ cross-cluster edges, sending $\Theta(1/n^2)$ of the $(u,v)$-flow over each edge. Thus, each cross cluster edge carries $\Theta(n^2 \times 1/n^2) = \Theta(1)$ flow, and the constant can be adjusted such that the unit capacity constraint is satisfied. Further note that the flow between each pair of nodes $(u, w) \in V_1 \times V_1$ is $\Theta(n \times q/n) = \Theta(q)$. In our regime, $q < p / \log n < \frac{d}{n \log n}$. As we already showed in Lemma~\ref{lemma:peakthput}, for a random regular graph, throughput (even for the complete demand graph) is $\Theta(\frac{d}{n \log n})$, and hence this flow is feasible. The same argument applies to internal flow in $V_2$ where flow from each cross-cluster edge is split again to the destinations.

Lastly, note that $T(q=q^*) = \Theta(q^*)= c_1\cdot p\frac{\log d}{\log n}=c_3\cdot d\frac{\log d}{n \log n} = \Theta(T^*)$
\end{proof}

Thus far, we have shown that for $q < q^*$, throughput $T(q) = \Theta(q)$ and $T(q=q^*) = \Theta(T^*)$, \ie within constant factor of the peak throughput. The following lemma will establish that $T(q > q^*) = \Theta(T^*)$ and thus prove our result.

\begin{lemma}
For $q > q^*$, $T(q)$ is within a constant factor of the peak throughput $T^*$.
\end{lemma}
\begin{proof}

First, we note than when $q > p$, the graph is an optimal bipartite expander and thus throughput is within a constant factor of $T(p=q) = \Theta(T^*)$~\cite{ellis}. When $q^* < q < p$, as we increase $q$, $p$ does not change by more than a factor of $2$. Thus, we can apply the same argument as Lemma~\ref{lemma:mainlemma} to route $\Theta(T^*)$ flow: clearly, increasing $q$ does not decrease flow, and $p$ changing by a constant factor only reduces it by a constant factor at most.
\end{proof}

\begin{proof} [Of Theorem \ref{thm:throughput}]
The above three lemmata directly imply the theorem for the demand graph $H=K_{V_1,V_2}$.  Note further that random permutation traffic demands $P$ can be routed within $H$ at a constant factor lower flow throughput.  Specifically, for each flow $v\to w$ between two nodes in the same cluster $V_1$ in $P$, we can split the flow $v\to w$ into $n/2$ subflows, from $v$ to each node in $V_2$ and from there to $w$.  After handling the other types of traffic (within cluster $V_2$ and across clusters) similarly, this produces a bipartite demand graph with a constant factor larger demands than $H$.  Hence, the theorem is concluded.
\end{proof}

%% file: comparison.tex
\section{Improving VL2}
\label{sec:vl2}

\begin{figure*}
\centering
\subfigure[]{ \label{fig:vl2:vl2comp}\includegraphics[width=2.11in]{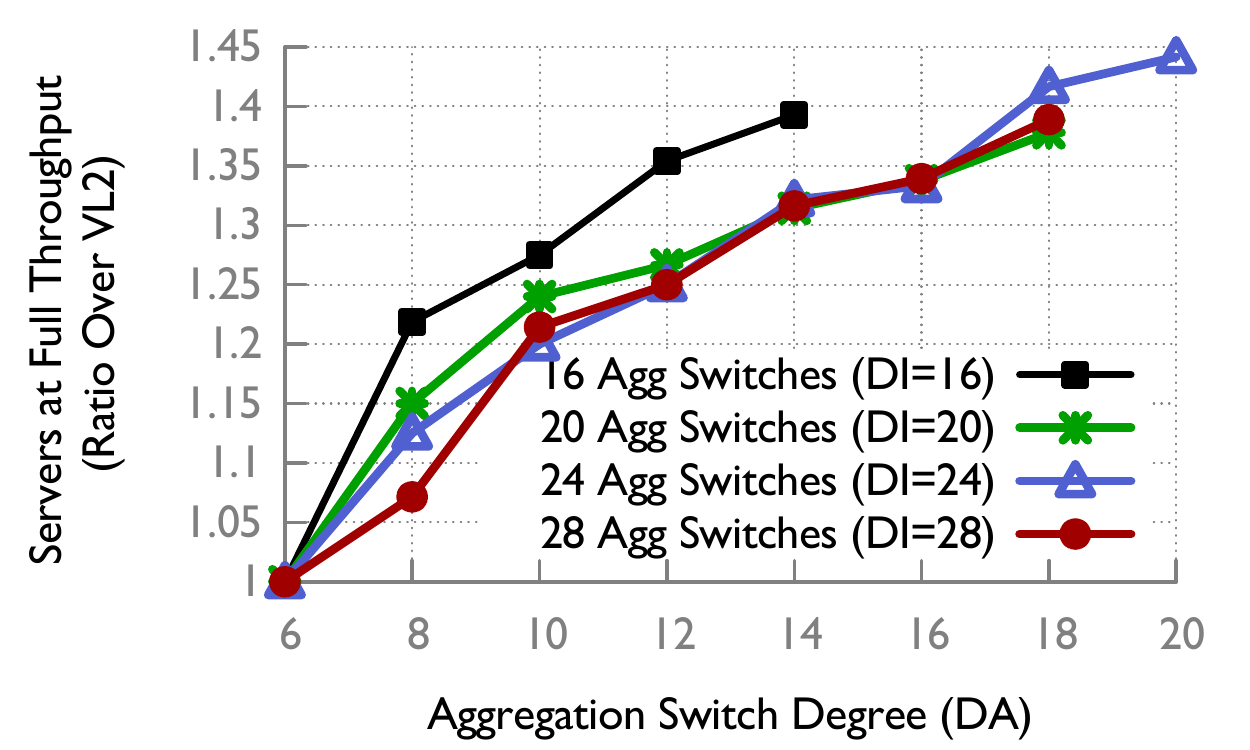}}
\subfigure[]{ \label{fig:vl2:vl2matrices}\includegraphics[width=2.11in]{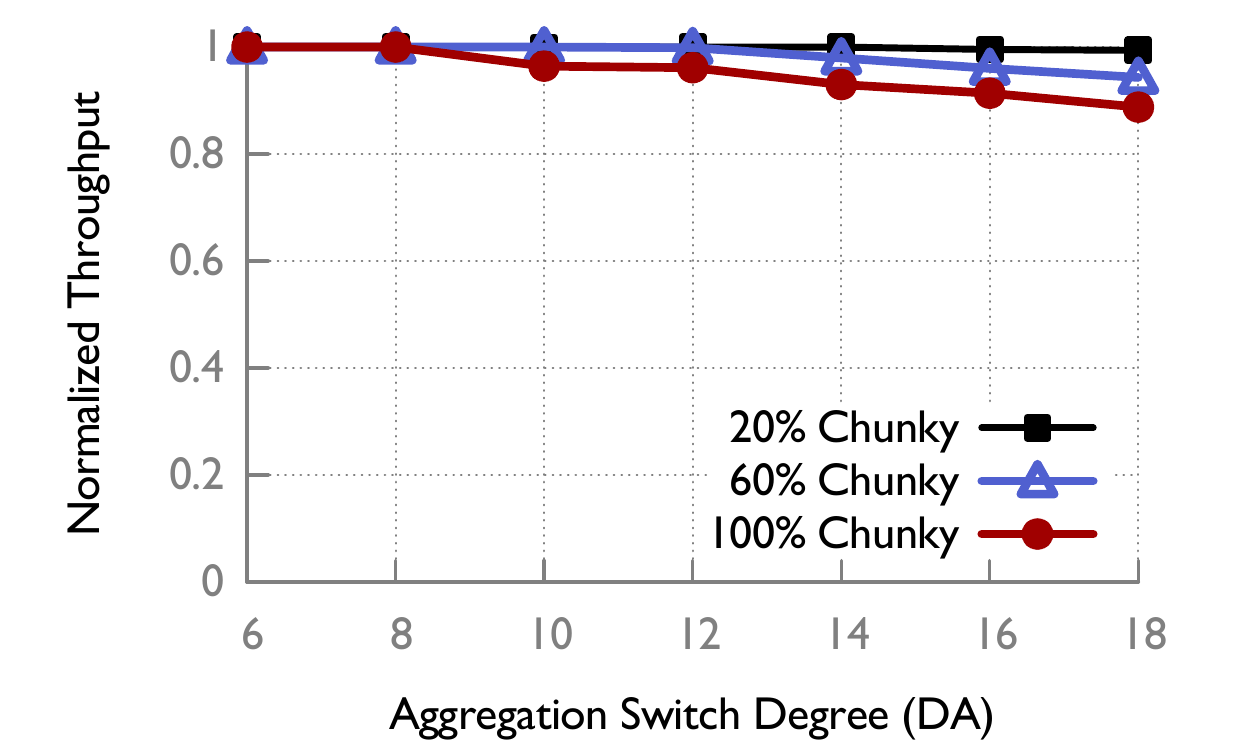}}
\subfigure[]{ \label{fig:vl2:vl2comp_stride}\includegraphics[width=2.11in]{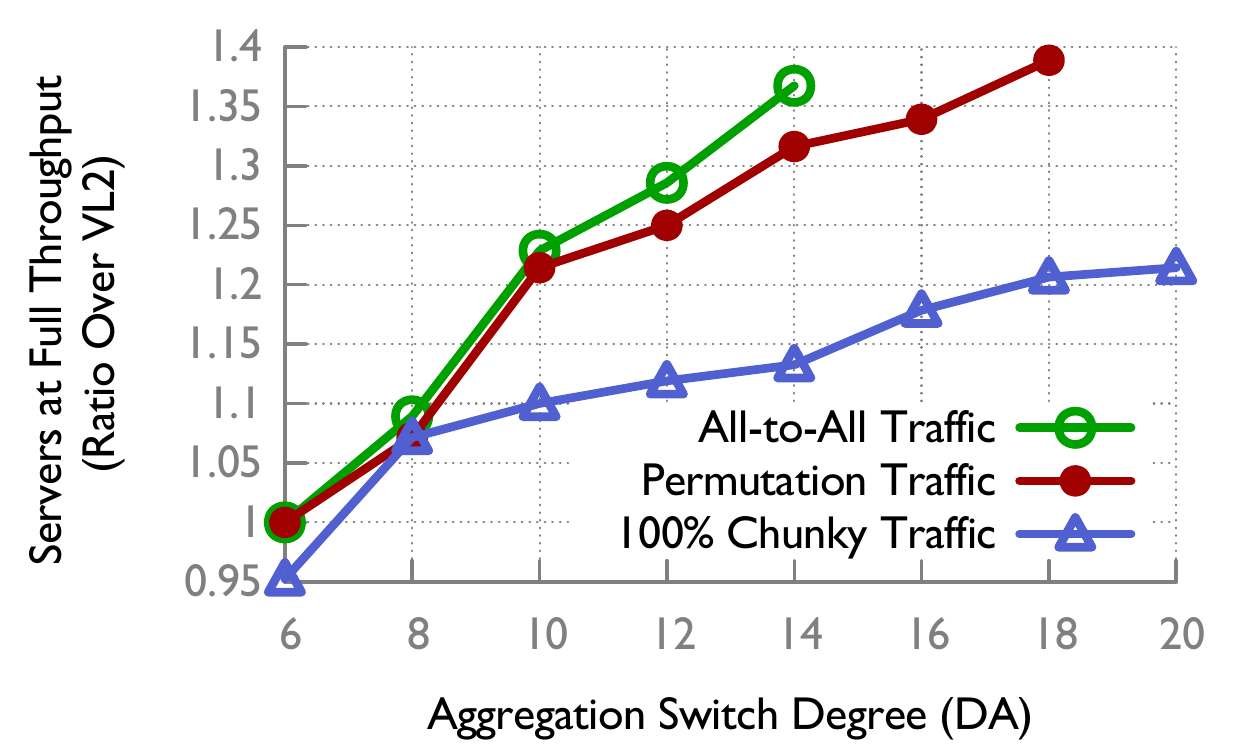}}
\caption{\small \em Improving VL2: (a) The number of servers our topology supports in comparison to
VL2 by rewiring the same equipment; (b) Throughput under various chunky traffic patterns; and (c) The
number of servers our topology can support in comparison to VL2 when we require it to achieve full
throughput for all-to-all traffic, permutation traffic, and chunky traffic.}
\label{fig:vl2}
\end{figure*}

In this section, we apply the lessons learned from our experiments and analysis to
improve upon a real world topology. Our case study uses the VL2~\cite{vl2} topology 
deployed in Microsoft's data centers. VL2 incorporates heterogeneous line-speeds 
and port-counts and thus provides a good opportunity for us to test our design ideas.

\cut{the one past proposal that incorporates significant heterogeneity\footnote{Note
that LEGUP was already analyzed in Jellyfish~\cite{legup} and shown to have significantly
lower throughput due to being restricted to Clos structures.}.}

\paragraphb{VL2 Background:} VL2~\cite{vl2} uses three types of switches: top-of-racks (ToRs), aggregation
switches, and core switches. Each ToR is connected to $20$ $1$GbE servers, and has $2$ $10$GbE
uplinks to different aggregation switches. The rest of the topology is a full bipartite 
interconnection between the core and aggregation switches. If aggregation switches have $DA$
ports each, and core switches have $DI$ ports each, then such a topology supports $\frac{DA.DI}{4}$
ToRs at full throughput.

\paragraphb{Rewiring VL2:} 
As results in \S\ref{subsec:heterports} indicate, connecting ToRs to only aggregation switches, instead of distributing
their connectivity across all switches is sub-optimal. Further, the results
on the optimality of random graphs in \S\ref{sec:jellyfish} imply further gains from using
randomness in the interconnect as opposed to VL2's complete bipartite interconnect. In line with these
observations, our experiments show significant gains obtained by modifying VL2.

In modifying VL2, we distribute the ToRs over aggregation and core switches in
proportion to their degrees. We connect the remaining ports uniform randomly. To measure
our improvement, we calculate the number of ToRs our topology can support at full throughput
compared to VL2. By `supporting at full throughput', we mean observing
full $1$ Gbps throughput for each flow in random permutation traffic across each of $20$ runs. We obtain the largest number of 
ToRs supported at full throughput by doing a binary search.
As Fig.~\ref{fig:vl2:vl2comp} shows, we gain as much as a $43\%$ improvement in
the number of ToRs (equivalently, servers) supported at full throughput at the largest size.
Note that the largest size we evaluated is fairly small -- $2$,$400$ servers for VL2 -- \ankitblue{and
our improvement increases with the network's size.}
	

%% file: practice.tex
\section{In Practice}
\label{sec:practice}
In this section, we address two practical concerns: (a) performance with a more 
diverse set of traffic matrices beyond the random permutations we have used so far; and 
(b) translating the flow model to packet-level throughput without changing the results significantly.


\subsection{Other Traffic Matrices}
\label{subsec:vl2:matrices}

We evaluate the throughput of our VL2-like topology under other traffic matrices besides
random permutations. For these experiments, we use the topologies corresponding to the `$28$ Agg Switches ($DI$$=$$28$)'
curve in Fig.~\ref{fig:vl2:vl2comp}. (Thus, by design, the throughput for random permutations is expected,
and verified, to be $1$.) In addition to the random permutation, we test the following other traffic 
matrices: (a) All-to-all: where each server communicates with every other server; 
and (b) $x\%$ Chunky: where each of $x\%$
of the network's ToRs sends all of its traffic to \emph{one} other ToR in this set (\ie a ToR-level
permutation), while the remaining $(100-x)\%$ ToRs are engaged in a server-level random 
permutation workload among themselves.

Our experiments showed that using the network to interconnect the same number of servers as in our earlier tests with random permutation traffic, full throughput is still achieved for all but the chunky traffic 
pattern. In Fig.~\ref{fig:vl2:vl2matrices}, we
present results for $5$ chunky patterns. Except when a majority of the network is 
engaged in the chunky pattern, throughput is within a few percent of full throughput.
We note that $100\%$ Chunky is a \ankitblue{hard to route traffic pattern} which is easy 
to avoid. Even assigning applications to servers randomly will ensure that the 
probability of such a pattern is near-zero. 

Even so, we repeat the experiment from Fig.~\ref{fig:vl2:vl2comp} where we had measured the
number of servers our modified topology supports at full throughput under random permutations.
In this instance, we require our topology to support full throughput under the $100\%$ Chunky 
traffic pattern. The results in Fig.~\ref{fig:vl2:vl2comp_stride} show that the gains are smaller,
but still significant, $22\%$ at the largest size, and increasing with size. It is also
noteworthy that all-to-all traffic is easier to route than both the other workloads.

\cut{We stress however, that VL2 is incurring a significant cost to handle an easy to avoid \ankitblue{corner-case}.
This discussion also points to a common misconception: that all-to-all 
traffic (something similar to which might manifest in the map-reduce shuffle phase) 
is hard to design for. (The VL2 paper uses it as a `stress test'.) All-to-all traffic
is, in fact, among the easy ones to design for -- each individual flow is small
because of bottlenecks at hosts, allowing for easier load balancing across the network. 
In contrast, for a pattern like Chunky, with each rack sending all its traffic to another rack,
a large number of flows compete for the same set of close-to-shortest paths between the 
two racks.
As Fig.~\ref{fig:vl2:vl2comp_stride} also shows, if we were only required to design for the 
all-to-all traffic pattern to achieve full throughput, our gains would be even larger than 
for random permutations. (Unfortunately, simulations with all-to-all traffic do not scale to
larger sizes.)}

\subsection{From Flows to Packets}
\label{subsec:routing}

Following the method used by Jellyfish~\cite{jellyfish}, we use Multipath TCP (MPTCP~\cite{mptcp})
in a packet level simulation to test if the throughput of our modified VL2-like topology
is similar to what flow simulations yield. We use MPTCP with the shortest paths, using
as many as $8$ MPTCP subflows. The results in Fig.~\ref{fig:routing} show that throughput
within a few percent ($6\%$ gap at the largest size) of the flow-level simulations is achievable. Note that we deliberately
oversubscribed the topologies so that the flow value was close to, but less than $1$. This 
makes sure that we measure the gap between the flow and packet models accurately --- if the
topology is overprovisioned, then even inefficient routing and congestion control may possibly 
yield close to full throughput.

\begin{figure}
\centering
\includegraphics[width=2.3in]{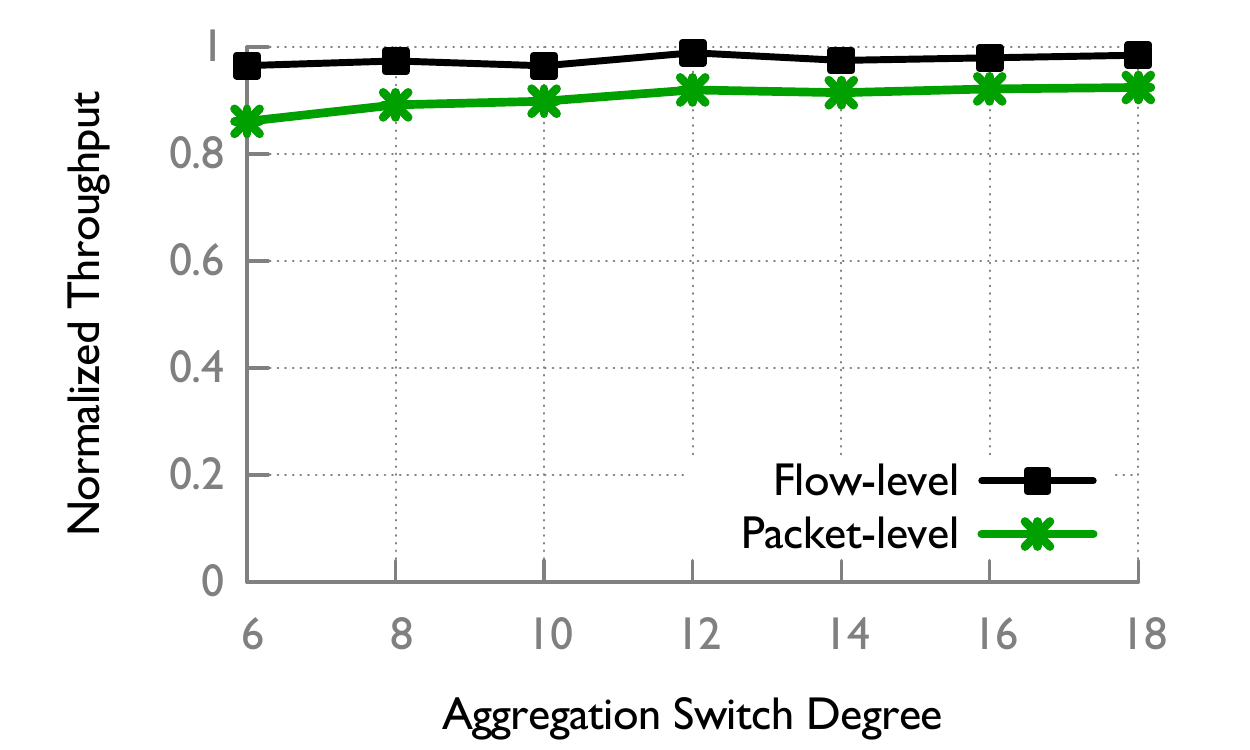}
\caption{\small \em Packet level simulations of random permutation traffic over our topology 
show that throughput within a few percent of the optimal flow-level throughput can be achieved
using MPTCP over the shortest paths.}
\label{fig:routing}
\end{figure}



%% file: discussion.tex
\section{Discussion}
\label{sec:future}

\paragraphb{Why these traffic matrices?} \ankitr{In line with the design objective of hosting arbitrary applications at high throughput, the approach we have taken is to study \emph{difficult} traffic matrices, rather than TMs specific to particular environments.} We show in~\cite{sigmetricsdraft} that an all-to-all workload bounds performance under any workload within a factor of $2$. As such, testing this TM is more useful than any other specific, arbitrary choice. In addition, we evaluate other traffic matrices which are even harder to route (Fig.~\ref{fig:vl2:vl2comp_stride}). Further, our code is available~\cite{codelink}, and is easy to augment with arbitrary traffic patterns to test. 

\paragraphb{What about latency?} We include a rigorous analysis of latency in terms of path length (Fig.~\ref{fig:optimal:pl}, \ref{fig:optimal:pl2}), showing that average shortest path lengths are close to optimal in random graphs. Further, Jellyfish~\cite{jellyfish} showed that even worst-case path length (diameter) in random graphs is smaller than or similar to that in fat-trees. Beyond path length, latency depends on the transport protocol's ability to keep queues small. In this regard, we note that techniques being developed for low latency transport (such as DCTCP~\cite{dctcp}, HULL~\cite{hull}, pFabric~\cite{pfabric}) are topology agnostic.

\paragraphb{But randomness?!} `Random' $\nRightarrow$ `inconsistent performance': the standard deviations in throughput are $\sim$$1\%$ of the mean (and even smaller for path length). Also, by `maximizing the minimum flow' to measure throughput, we impose a strict definition of fairness, eliminating the possibility of randomness skewing results across flows. Further, Jellyfish~\cite{jellyfish} showed that random graphs achieve flow-fairness comparable to fat-trees under a practical routing scheme. Simple and effective physical cabling methods were also shown in~\cite{jellyfish}.

\paragraphb{Limitations:} While we have presented here foundational results on the design of both homogeneous and heterogeneous topologies, many interesting problems remain unresolved, including:
(a) a non-trivial upper bound on the throughput of heterogeneous networks; 
(b) theoretical support for our \S\ref{subsec:heterports} result on server distribution; and
(c) generalizing our results to arbitrarily diverse networks with multiple switch types.

Lastly, we note that this work does not incorporate functional constraints such as those imposed by middleboxes, for instance, in its treatment of topology design.

%% file: conclusion.tex
\section{Conclusion}

Our result on the near-optimality of random graphs for homogeneous
network design implies that homogeneous topology design may be reaching 
its limits, particularly when uniformly high throughput is desirable. 
The research community should perhaps focus its efforts on other aspects of the problem, 
such as joint optimization with cabling, or topology design for
specific traffic patterns (or bringing to practice research proposals 
on the use of wireless and/or optics for flexible networks that adjust 
to traffic patterns), or improvements to heterogeneous network 
design beyond ours.

Our work also presents the first systematic approach to the design
of heterogeneous networks, allowing us to improve upon a deployed data
center topology by as much as $43\%$ even at
the scale of just a few thousand servers, with this improvement
increasing with size. In addition, we further the understanding 
of network throughput by showing how cut-size, path length, and 
utilization affect throughput.

While significant work remains in the space of designing and analyzing
topologies, this work takes the first steps away from the myriad point
solutions and towards a theoretically grounded approach to the problem.

%% file: ack.tex
\section*{Acknowledgments}
We would like to thank Walter Willinger and our anonymous NSDI $2014$ reviewers 
for their valuable suggestions. We gratefully acknowledge the support of Cisco Research Council Grant $573665$.
Ankit Singla was supported by a Google PhD Fellowship.